\documentclass[final,5p,times,twocolumn]{elsarticle}
\usepackage{graphicx}      
\usepackage{natbib}        
\usepackage{amsmath,amsthm}
\usepackage{amssymb}
\usepackage{bm}

\usepackage[noend]{algpseudocode}
\usepackage{algorithmicx,algorithm}

\usepackage{amssymb}
\newtheorem{thm}{Theorem}
\newtheorem{cor}{Corollary}

\newtheorem{mark_}{Remark}
\newdefinition{definition}{Definition}

\newtheorem{ass}{Assumption}

\def\bf{\bm}

\journal{Journal of Process Control}

\begin{document}

\begin{frontmatter}



\title{Analysis and tuning of a three-term DMC}

\author[label1]{Yun Zhu}
\author[label2]{Kangkang Zhang}
\author[label1]{Yucai Zhu}
\author[label1]{Jinming Zhou}
\affiliation[label1]{organization={State Key Laboratory of Industrial Control Technology, College of Control Science and Engineering, Zhejiang university},
             city={Hangzhou},
             postcode={310027},
             country={China}}
\affiliation[label2]{organization={Hangzhou Tai-Ji Control Ltd.},
             city={Hangzhou},
             postcode={310000},
             country={China}}
\tnotetext[mytitlenote]{This work is supported by National Natural Science Foundation of China [grant number U1809207] and National Key R\&D Program of China [grant number 2017YFA0700300].}

\begin{abstract}
Most MPC (Model Predictive Control) algorithms used in industries and studied in the control academia use a two-term QP (quadratic programming), where the first term is the weighted norm of the output errors, and the second term is that of the input increments. In this work, a DMC (Dynamic Matrix Control) algorithm that uses three-term QP is studied, where the third term is the weighted norm of the output increments. In the analysis, a relationship between the three-term DMC and the two-term DMC is established; based on that, the closed-loop response curves are derived. Based on the analysis, two controller tuning procedures are developed for the three-term DMC, one for closed-loop step response and one for disturbance reduction. Finally, it will be proven that the three-term DMC can achieve a higher performance and robustness than the two-term DMC can. Simulation studies are used to demonstrate the findings and the tuning methods.
\end{abstract}

\begin{keyword}
Model predictive control (MPC), dynamic matrix control (DMC), three-term DMC, tuning procedure, control performance
\end{keyword}

\end{frontmatter}

\section{Introduction}
The idea of model predictive control (MPC) was first introduced in 1960s by Zadeh, Propoi and Rafal (\citep{Zadeh1962, Propoi1963, Rafal1968}. Refinements of MPC algorithms were subsequently made by researchers who applied them to real-world problems, as seen in the work of Richalet et al. and Cutler et al. (\citep{Cutler1980,Richalet1978}). MPC is widely applied for its ability to treat multi-variable and complex control problems and handle process constraints. The control community started to study MPC much later, partly due to its success in process industries. Many researchers focus on the stability of MPC systems. Stability theorems are developed by invoking terminal cost, and terminal constraint set \cite{Mayne2000}. Robustness concerning model uncertainty is also studied, and min-max MPC and Tube MPC were proposed \citep{Campo1987, Limon2010}. However, in process industries, the first concern of MPC systems is the control performance rather than closed-loop stability because it is the control performance that brings economic benefit to the operating company. The control performance usually refers to set point tracking and disturbance reduction. A high-performance MPC system has a fast and smooth setpoint tracking and/or a small output variance, which serves to increase the the safety and profit margin of the plant operation. The second concern of the process industries is the easiness of MPC tuning. In this work, instead of studying the stability of MPC systems, the authors try to find ways to enhance the control performance of MPC systems and propose user-friendly tuning methods.

One important aspect of MPC is the choice of the loss function used to quantify the performance of the control action. Different loss functions can lead to different control behaviors and trade-offs between different objectives. For instance, in power electronics, researchers commonly use the ${\ell _1}$-norm instead of the squared ${\ell _2}$-norm in the objective function of MPC with reference tracking to achieve fast control and stabilization, as reported in \citep{Karamanakos2018, FEHER2020242}. Additionally, to ensure closed-loop stability, MPC set-up based on $\infty $ norm cost function has been proposed in \citep{lazar2007discrete}.

In a traditional MPC scheme, the objective function (${\ell _2}$-norm) to be minimized penalizes the norm of the output error and the norm of the input increment using weighting matrices $\bm{Q}$ and $\bm{R}$ (\cite{Qin2003}). Here, this scheme will be referred as two-term MPC. In the last two decades, several researchers have proposed a three-term MPC scheme by adding a weighted norm of the output increment into the objective function (\citep{Chen1996, Zhang2008}). In these two papers, the authors proposed a three-term MPC scheme with some simulations, and no analysis was given of its potential advantages. Gomma \cite{Gomma2005} claimed that this scheme could increase the robustness of the system. Mollov \cite{Mollov2004} mentioned that one could adopt the three-term scheme to limit the variation in the predicted output rather than by hard constraints to avoid the infeasibility of the optimization problem. However, within the control community, the three-term MPC has received little attention. Its advantages and good properties have not been explored yet.

DMC (Dynamic Matrix Control) can be seen as a type of MPC that uses a specific model structure and optimization approach, which is a widely used and perhaps most successful control strategy in process industries. In this work, the authors will study the three-term DMC scheme in two aspects. First, the authors established the relationship between the three-term DMC and the two-term DMC, and derived the formula for calculating the closed-loop response curves. Based on these results, two practical and user-friendly DMC tuning methods are proposed. Finally, it will be shown that the three-term DMC outperforms the two-term DMC in control performance and robustness.

\emph{Contributions and limitations of the work}

This work is primarily constructive rather than analytical. The authors establish a relationship between the three-term DMC and the two-term DMC, and based on this, develop two user-friendly MPC tuning methods. The authors also propose a method for comparing two control methods and prove that the three-term DMC is superior to the two-term DMC in both performance and robustness. There are limitations in this work. The authors cannot prove the asymptotic stability of both the two-term DMC and the three-term DMC (in industrial applications, the stability of DMC controlled systems are verified using simulations); and the analysis in the work cannot handle process constraints. The authors ask their readers for forgiveness for their weakness in theoretical analysis.

The rest of the paper is structured as follows: In section 2, the three-term MPC algorithm is introduced; in Section 3, the equivalent form of the three-term MPC to the two term MPC is derived and the formula for closed-loop response curves is developed; in section 4, two tuning methods are proposed; Section 5 shows the advantages of the three-term MPC scheme in terms of control performance and robustness; Section 6 is the conclusion.

\section{Description of DMC algorithm}
\subsection{Two-term DMC algorithm}
DMC is an industrial model predictive control algorithm developed by Cutler et.al. \cite{Cutler1980} where the process model is given in the form of process step responses which is called dynamic matrix. Consider a multi-variable process with $m$ inputs and $p$ outputs.
\begin{equation}\label{eq1}
{\bm{Y}}(k) = {\bm{G}}({q^{ - 1}}){\bm{U}}(k)
\end{equation}
where ${\bm{G}}({q^{ - 1}})$ is a $p \times m$ transfer function matrix, ${\bm{Y}}(k)$ and ${\bm{U}}(k)$ are the output vector and the input vector of the process, ${q^{ - 1}}$ denotes unit delay operator, ${q^{ - 1}}{\bm{U}}(k) = {\bm{U}}(k - 1)$. Let ${a_{ij}}(k)$ represents the step response coefficient from the input $i$ to the output $j$ at time sample $k$, where $i = 1,...,p;j = 1,...,m;k = 1,...,N$. Denote $N$ as the horizon of dynamics, $P$ as the prediction horizon, and $M$ as the control horizon. Then the dynamic matrix $A$ is a $pP \times mM$ matrix which consists of the step response coefficients ${a_{ij}}(k)$.
\begin{equation}\label{eq2}
{\bm{A}} = \left[ {\begin{array}{*{20}{c}}
{{{\bm{A}}_{11}}}& \cdots &{{{\bm{A}}_{1m}}}\\
 \vdots & \ddots & \vdots \\
{{{\bm{A}}_{p1}}}& \cdots &{{{\bm{A}}_{pm}}}
\end{array}} \right]
\end{equation}
\begin{equation}\label{eq3}
{{\bm{A}}_{ij}} = \left[ {\begin{array}{*{20}{c}}
{{a_{ij}}(1)}&{}&0\\
 \vdots & \ddots &{}\\
{{a_{ij}}(M)}& \cdots &{{a_{ij}}(1)}\\
 \vdots &{}& \vdots \\
{{a_{ij}}(P)}& \cdots &{{a_{ij}}(P - M + 1)}
\end{array}} \right]
\end{equation}
The prediction formula can be expressed as
\begin{equation}\label{eq4}
{{\bm{Y}}_P}(k) = {{\bm{Y}}_{P0}}(k) + {\bm{A}}\Delta {\bm{U}}(k)
\end{equation}
where
\begin{equation}\label{eq5}
{{\bm{Y}}_P}(k) = {\left[ {\begin{array}{*{20}{c}}
{{{\bm{y}}_P}_1^T(k)}& \cdots &{{{\bm{y}}_P}_p^T(k)}
\end{array}} \right]^T}
\end{equation}
${{\bm{Y}}_P}(k)$ is the vector of the predicted outputs from sample time $k+1$ to $k+P$ packed on top of each other with that of the output $i$ denoted as
\begin{equation}\label{eq6}
{{\bm{y}}_{Pi}}(k) = {\left[ {\begin{array}{*{20}{c}}
{{y_{Pi}}(k + 1)}& \cdots &{{y_{Pi}}(k + P)}
\end{array}} \right]^T}
\end{equation}
Similarly,
\begin{equation}\label{eq7}
{{\bm{Y}}_N}(k) = {\left[ {\begin{array}{*{20}{c}}
{{{\bm{y}}_N}_1^T(k)}& \cdots &{{{\bm{y}}_N}_p^T(k)}
\end{array}} \right]^T}
\end{equation}
${{\bm{Y}}_N}(k)$ is the vector of the predicted outputs from sample time $k+1$ to $k+N$ packed on top of each other with that of the output $i$ denoted as
\begin{equation}\label{eq8}
{{\bm{y}}_{Ni}}(k) = {\left[ {\begin{array}{*{20}{c}}
{{y_{Ni}}(k + 1)}& \cdots &{{y_{Ni}}(k + N)}
\end{array}} \right]^T}
\end{equation}
The vector ${{\bf{Y}}_{P0}}(k)$ is the predicted outputs assuming the input ${\bm{U}}(k)$ remain the constant value at ${\bm{U}}(k-1)$.
\begin{equation}\label{eq9}
\Delta {\bm{U}}(k) = {\left[ {\begin{array}{*{20}{c}}
{\Delta {\bm{u}}_1^T(k)}& \cdots &{\Delta {\bm{u}}_m^T(k)}
\end{array}} \right]^T}
\end{equation}
The decision variable $\Delta {\bm{U}}(k)$ is the vector of the input increment from sample time $k+1$ to $k+M$ packed on top of each other with that of input $j$ denoted as
\begin{equation}\label{eq10}
\Delta {{\bm{u}}_j}(k) = {\left[ {\begin{array}{*{20}{c}}
{\Delta {u_j}(k)}& \cdots &{\Delta {u_j}(k + M - 1)}
\end{array}} \right]^T}
\end{equation}
The traditional DMC algorithm determines the future control moves ($\Delta {\bm{U}}(k)$) over the control horizon ($M$) to drive the model predicted outputs as closely as possible to the desired future trajectory over prediction horizon ($P$). The computation of DMC control actions is the solution of the following QP (quadratic programming)
\begin{equation}\label{eq11}
\mathop {\min }\limits_{\Delta {\bf{U}}(k)} {J_{2term}}(k) = \left\| {{\bf{W}}(k) - {{\bf{Y}}_P}(k)} \right\|_{\bf{Q}}^2 + \left\| {\Delta {\bf{U}}(k)} \right\|_{\bf{R}}^2
\end{equation}
where
\begin{equation}\label{eq12}
{\bm{W}}(k) = {\left[ {\begin{array}{*{20}{c}}
{{\bm{w}}_1^T(k)}& \cdots &{{\bm{w}}_p^T(k)}
\end{array}} \right]^T}
\end{equation}
${{\bm{W}}(k)}$ is the vector of set point organized the same way as the vector of the predicted outputs with the ${i^{th}}$ set point vector denoted as
\begin{equation}\label{eq13}
{{\bm{w}}_i}(k) = {\left[ {\begin{array}{*{20}{c}}
{{w_i}(k + 1)}& \cdots &{{w_i}(k + P)}
\end{array}} \right]^T}
\end{equation}
The values in Eq. (\ref{eq13}) can be constant set points for the outputs and reference trajectories, such as first-order step responses. It is assumed that constant sequences are used as the reference trajectories unless specified otherwise.
${\bm{Q}}$ and ${\bm{R}}$ represent the output weighting matrix and the input incremental weighting matrix.
\begin{equation}\label{eq14}
{\bm{Q = }}\left[ {\begin{array}{*{20}{c}}
{{{\bm{Q}}_1}}&{}&{}\\
{}& \ddots &{}\\
{}&{}&{{{\bm{Q}}_p}}
\end{array}} \right],{\rm{ }}{\bm{R = }}\left[ {\begin{array}{*{20}{c}}
{{{\bm{R}}_1}}&{}&{}\\
{}& \ddots &{}\\
{}&{}&{{{\bm{R}}_m}}
\end{array}} \right]
\end{equation}
where
\begin{equation}
{{\bm{Q}}_i} = {\left[ {\begin{array}{*{20}{c}}
  0&{}&{}&{}&{}&{} \\
  {}& \ddots &{}&{}&{}&{} \\
  {}&{}&0&{}&{}&{} \\
  {}&{}&{}&{{q_i}}&{}&{} \\
  {}&{}&{}&{}& \ddots &{} \\
  {}&{}&{}&{}&{}&{{q_i}}
\end{array}} \right]_{P \times P}},{{\bm{R}}_j} = {\left[ {\begin{array}{*{20}{c}}
  {{r_j}}&{}&{} \\
  {}& \ddots &{} \\
  {}&{}&{{r_j}}
\end{array}} \right]_{M \times M}}
\end{equation}
where $\bm{Q}_i$ and ${{\bm{R}}_j}$ are the weighting matrixes for output $i$ and input $j$, and ${q_i} \in [0,\infty ),{r_j} \in [0,\infty )$.
In the diagonal of $\bm{Q}_i$, zero value is given to treat delay and reverse characteristics. The process has at least one sample delay by default. Accordingly, if the $i^{th}$ output has $d_i$ delay samples, there should be $(d_i-1)$ zeros in the diagonal of $\bm{Q}_i$. If the optimization has no constraint or does not trigger any constraint, the optimization in Eq. (\ref{eq11}) has an analytical solution:
\begin{equation}\label{eq15}
\Delta {\bm{U}}(k) = {\left( {{{\bm{A}}^T}{\bm{QA}} + {\bm{R}}} \right)^{ - 1}}{{\bm{A}}^T}{\bm{Q}}\left( {{\bm{W}}(k) - {{\bm{Y}}_{P0}}(k)} \right)
\end{equation}
where $\Delta {\bm{U}}(k)$ contains $M$ steps. The controller only takes the first step increment $\Delta {{\bm{U}}_1}(k)$ for actual control, namely:
\begin{equation}\label{eq16}
\Delta {{\bm{U}}_1}(k) = {\bm{L}}\Delta {\bm{U}}(k)
\end{equation}
\begin{equation}\label{eq17}
{\bm{L}} = {\left[ {\begin{array}{*{20}{c}}
1&0& \cdots &0&{}&{}&{}&0&{}\\
{}&{}&{}&{}& \ddots &{}&{}&{}&{}\\
{}&0&{}&{}&{}&1&0& \cdots &0
\end{array}} \right]_{m \times Mm}}
\end{equation}
at the next sample interval, a similar optimization is performed. This is so-called "Moving horizon optimization" strategy.
\subsection{Three-term DMC algorithm}
Previous studies on the three-term model predictive control (MPC) primarily used the generalized predictive control (GPC) algorithm \citep{Chen1996, Zhang2008}. However, the three-term scheme is applicable to all MPC algorithms. In this study, we utilize the dynamic matrix control (DMC) algorithm introduced in the last section. The loss function of the three-term DMC consists of three terms, with the third term being the weighted norm of the output increment. The overall loss function is formulated as follows:
\begin{equation}\label{eq23}
\begin{array}{l}
\mathop {\min }\limits_{\Delta {\bf{U}}(k)} {J_{3term}}(k) = \left\| {{\bf{W}}(k) - {{\bf{Y}}_P}(k)} \right\|_{\bf{Q}}^2 + \left\| {\Delta {\bf{U}}(k)} \right\|_{\bf{R}}^2 + \left\| {\Delta {{\bf{Y}}_P}(k)} \right\|_{\bf{S}}^2\\
 = \left\| {{\bf{W}}(k) - {{\bf{Y}}_0}(k) - {\bf{A}}\Delta {\bf{U}}(k)} \right\|_{\bf{Q}}^2 + \left\| {\Delta {\bf{U}}(k)} \right\|_{\bf{R}}^2\\
 + \left\| {{{\bf{T}}_2}\left( {{{\bf{Y}}_0}(k) + {\bf{A}}\Delta {\bf{U}}(k)} \right) - {{\bf{T}}_3}{\bf{Y}}(k)} \right\|_{\bf{S}}^2
\end{array}
\end{equation}
Here, ${\Delta {{\bm{Y}}_P}(k)}$ represents the increment of ${{{\bm{Y}}_P}(k)}$, which is defined as the predicted output increment from sample time $k+1$ to $k+P$ packed on top of each other with that of input $j$. Specifically, one obtains:
\begin{equation}\label{eq24}
\Delta {{\bm{Y}}_P}(k) = {{\bm{T}}_2}{{\bm{Y}}_P}(k) - {{\bm{T}}_3}{\bm{Y}}(k)
\end{equation}
In this equation, ${\bm{Y}}(k)$ represents the output at the last sample, and $\Delta {{\bm{Y}}_P}(k)$ is a vector of predicted output increments. The predicted output increments are given by:
\begin{equation}\label{eqYy}
\Delta {{\bf{Y}}_P} = {\left[ {\begin{array}{*{20}{c}}
{\Delta {\bf{y}}_{P,1}^T}& \ldots &{\Delta {\bf{y}}_{P,p}^T}
\end{array}} \right]^T}
\end{equation}
where $\Delta {{\bf{y}}_{P,i}}$ denote the predictive increment of the $i^th$ output
\begin{equation}\label{eq25}
\Delta {{\bf{y}}_{P,i}} = {\left[ {\begin{array}{*{20}{c}}
{{y_{P,i}}(k + 1) - {y_i}(k)}& \ldots &{{y_{P,i}}(k + P) - {y_{P,i}}(k + P - 1)}
\end{array}} \right]^T}
\end{equation}
${{\bm{T}}_2}$ is a difference matrix,
\begin{equation}
{{\bm{T}}_2} = {\left[ {\begin{array}{*{20}{c}}
  {{t_2}}&{}&{} \\
  {}& \ddots &{} \\
  {}&{}&{{t_2}}
\end{array}} \right]_{pP \times pP}},{{\bm{T}}_3} = {\left[ {\begin{array}{*{20}{c}}
  {{t_3}}&{}&{} \\
  {}& \ddots &{} \\
  {}&{}&{{t_3}}
\end{array}} \right]_{pP \times p}}
\end{equation}
\begin{equation}\label{eq26}
{t_2} = {\left[ {\begin{array}{*{20}{c}}
  1&0&{}&{}&{} \\
  { - 1}&1&{}&{}&{} \\
  {}&{}& \ddots &{}&{} \\
  {}&{}&{}&1&0 \\
  {}&{}&{}&{ - 1}&1
\end{array}} \right]_{P \times P}},{\text{  }}{t_3} = {\left[ {\begin{array}{*{20}{c}}
  1 \\
  0 \\
   \vdots  \\
  0 \\
  0
\end{array}} \right]_{P \times 1}}
\end{equation}
${\bm{S}}$ is the output incremental weighting matrix, which is also a diagonal matrix, namely,
\begin{equation}\label{eq27}
{\bm{S = }}{\left[ {\begin{array}{*{20}{c}}
  {{{\bm{S}}_1}}&{}&{} \\
  {}& \ddots &{} \\
  {}&{}&{{{\bm{S}}_p}}
\end{array}} \right]_{pP \times pP}}
\end{equation}
where
\begin{equation}
{{\bm{S}}_i} = {\left[ {\begin{array}{*{20}{c}}
  0&{}&{}&{}&{}&{} \\
  {}& \ddots &{}&{}&{}&{} \\
  {}&{}&0&{}&{}&{} \\
  {}&{}&{}&{{s_i}}&{}&{} \\
  {}&{}&{}&{}& \ddots &{} \\
  {}&{}&{}&{}&{}&{{s_i}}
\end{array}} \right]_{P \times P}}
\end{equation}
where ${{\bm{S}}_i}$ is the weighting matrix for the increment of output $i$, and ${s_i} \in [0,{\rm{ }}\infty )$. In the diagonal of $\bm{S}_i$, zero value is set to treat delay and reverse characteristic. Taking the derivative of ${{J_{3term}}(k)}$ with respect to ${\Delta {\bm{U}}(k)}$, one obtains
\begin{equation}\label{eq28}
\begin{array}{l}
\frac{{d{J_{3term}}}}{{d\Delta {\bf{U}}(k)}} =  - 2{{\bf{A}}^T}{\bf{Q}}\left[ {{\bf{W}}(k) - {{\bf{Y}}_{P0}}(k) - {\bf{A}}\Delta {\bf{U}}(k)} \right]\\
 + 2{\bf{R}}\Delta {\bf{U}}(k) + 2{{\bf{A}}^T}{\bf{T}}_2^T{\bf{S}}\left[ {{{\bf{T}}_2}\left( {{{\bf{Y}}_{P0}}(k) + {\bf{A}}\Delta {\bf{U}}(k)} \right) - {{\bf{T}}_3}{\bf{Y}}(k)} \right]
\end{array}
\end{equation}
If the optimization has no constraints or does not trigger the constraints, the optimization in Eq. (\ref{eq23}) has an analytical solution:
\begin{equation}\label{eq29}
\begin{gathered}
  \Delta {\bm{U}}(k) = {\left( {{{\bm{A}}^T}{\bm{QA}} + {\bm{R}} + {{\bm{A}}^T}{\bm{T}}_2^T{\bm{S}}{{\bm{T}}_2}{\bm{A}}} \right)^{ - 1}}\left( {{{\bm{A}}^T}{\bm{Q}} \cdot } \right. \hfill \\
  \left( {{\bm{W}}(k) - {{\bm{Y}}_{P0}}(k)} \right) - \left. {{{\bm{A}}^T}{\bm{T}}_2^T{\bm{S}}{{\bm{T}}_2}{{\bm{Y}}_{P0}}(k) + {{\bm{A}}^T}{\bm{T}}_2^T{\bm{S}}{{\bm{T}}_3}{\bm{Y}}(k)} \right) \hfill \\
\end{gathered}
\end{equation}
The above description is an open-loop three-term DMC algorithm. The feedback mechanism is the same as that of the two-term DMC algorithm.
\subsection{Three-term DMC algorithm with constraints}
The optimization can be formulated as minimizing the objective function ${J_{3term}}(k)$ with respect to the control input vector $\Delta {\bf{U}}(k)$, subject to certain constraints.
\begin{equation}\label{3termqp}
\begin{array}{l}
\mathop {\min }\limits_{\Delta {\bf{U}}(k)} {J_{3term}}(k) = \left\| {{\bf{W}}(k) - {{\bf{Y}}_P}(k)} \right\|_{\bf{Q}}^2 + \left\| {\Delta {\bf{U}}(k)} \right\|_{\bf{R}}^2 + \left\| {\Delta {{\bf{Y}}_P}(k)} \right\|_{\bf{S}}^2\\
s.t.{\quad \quad \quad \quad}{{\bf{Y}}_P}(k) = {{\bf{Y}}_{P0}}(k) + {\bf{A}}\Delta {\bf{U}}(k)\\
{\quad \quad \quad \quad \quad}{{\bf{U}}_{\min }} \le {\bf{U}} \le {{\bf{U}}_{\max }}\\
{\quad \quad \quad \quad \quad}\Delta {{\bf{U}}_{\min }} \le \Delta {\bf{U}} \le \Delta {{\bf{U}}_{\max }}\\
{\quad \quad \quad \quad \quad}{{\bf{Y}}_{\min }} \le {{\bf{Y}}_P} \le {{\bf{Y}}_{\max }}
\end{array}
\end{equation}
The given optimization problem can be represented by a quadratic programming (QP) problem, which is a classic form of optimization,
\begin{equation}
\begin{array}{l}
\mathop {\min }\limits_{\Delta {\bf{U}}(k)} {J_{3term}}(k) = \frac{1}{2}\Delta {{\bf{U}}^T}(k){\bf{\Phi }}\Delta {{\bf{U}}^T}(k) + {{\bf{f}}^T}\Delta {{\bf{U}}^T}(k)\\
s.t.{\quad \quad \quad \quad}\bf{\Omega} \Delta {\bf{U}}(k) \le {\bf{\omega }}
\end{array}
\end{equation}
where the positive semi-definite matrix ${\bf{\Phi }}$ is formulated as
\begin{equation}
{\bf{\Phi }} = {{\bf{A}}^T}{\bf{QA}} + {\bf{R}} + {{\bf{A}}^T}{\bf{T}}_2^T{\bf{S}}{{\bf{T}}_2}{\bf{A}}
\end{equation}
The vector ${\bf{f}}$ is formulated as
\begin{equation}
{\bf{f}} = {\left( {{{\bf{Y}}_{P0}}(k) - {\bf{W}}(k)} \right)^T}{\bf{QA}} + {\left( {{{\bf{T}}_2}{{\bf{Y}}_{P0}}(k) - {{\bf{T}}_3}{\bf{Y}}(k)} \right)^T}{\bf{S}}{{\bf{T}}_2}{\bf{A}}
\end{equation}
In addition to the objective function, the QP also includes constraints on the control input vector. These constraints are linear inequalities, represented by a matrix ${\bf{\Omega}}$ and a vector ${\bf{\omega}}$.
\begin{equation}
{\bf{\Omega }} = {\left[ {\begin{array}{*{20}{c}}
{\bf{B}}&{\bf{I}}&{\bf{A}}&{ - {\bf{B}}}&{ - {\bf{I}}}&{ - {\bf{A}}}
\end{array}} \right]^T}
\end{equation}
where the matrix ${\bf{B}}$ is a block diagonal matrix with each diagonal element being a positive constant ${\bf{b}}_i$,
\begin{equation}
{\bf{B}} = {\left[ {\begin{array}{*{20}{c}}
{{{\bf{b}}_1}}&{}&{}\\
{}& \ddots &{}\\
{}&{}&{{{\bf{b}}_m}}
\end{array}} \right]_{mM \times mM}},{{\bf{b}}_i} = {\left[ {\begin{array}{*{20}{c}}
1&{}&{}\\
 \vdots & \ddots &{}\\
1& \cdots &1
\end{array}} \right]_{M \times M}}
\end{equation}
$A$ represents the dynamic matrix and $\bf{\omega }$ is formulated as
\begin{equation}
{\bf{\omega }} = \left[ {\begin{array}{*{20}{c}}
{{{\bf{U}}_{\max }}(k) - {\bf{U}}(k - 1)}\\
{\Delta {{\bf{U}}_{\max }}(k)}\\
{{{\bf{Y}}_{\max }}(k) - {{\bf{Y}}_{P0}}(k)}\\
{ - \left( {{{\bf{U}}_{\max }}(k) - {\bf{U}}(k - 1)} \right)}\\
{ - \Delta {{\bf{U}}_{\max }}(k)}\\
{ - \left( {{{\bf{Y}}_{\min }}(k) - {{\bf{Y}}_{P0}}(k)} \right)}
\end{array}} \right]
\end{equation}
The use of the third term in the QP was introduced independently by the authors and it was motivated by increasing the response speed of the system without causing oscillations. It then turned out that the use of the three-term MPC was proposed long before this work (~\cite{Chen1996}, ~\cite{Zhang2008}). In this paper, without claiming the invention of the three-term MPC, the authors will reveal some nice properties that have not been shown before, and develop easy-to-use tuning methods for the three-term DMC.

\section{Analysis of the three-term DMC}
\subsection{Equivalent expressions in the two-term DMC}
In order to understand how the three-term DMC works, its relation to the two-term DMC is established in the following theorem:

\begin{thm}\label{thm1}
  The loss function of the three-term DMC given in Eq. (\ref{eq23}) can be equivalently written in a two-term form:
  \begin{equation}\label{3to2}
    \begin{gathered}
      {J_{{\text{3term}}}}(k) = \left\| {{{\bm{W}}^{'}}(k) - {{\bm{Y}}_P}(k)} \right\|_{{{\bm{Q}}^{'}}}^2 + \left\| {\Delta {\bm{U}}(k)} \right\|_{\bm{R}}^2 \hfill \\
    \end{gathered}
    \end{equation}
    with ${\bf{Q}}^{'}$ and $\bm{W}^{'}(k)$ defined analogously to ${\bf{Q}}$ and $\bm{W}(k)$:
    \begin{equation}
      {{\bf{Q}}^{'}} = {\left[ {\begin{array}{*{20}{c}}
        {{\bf{Q}}_1^{'}}&{}&{} \\
        {}& \ddots &{} \\
        {}&{}&{{\bf{Q}}_p^{'}}
      \end{array}} \right]_{pP \times pP}},
    \end{equation},
    \begin{equation}
      {\bm{W}^{'}}(k) = {\left[ {\begin{array}{*{20}{c}}
      {{\bm{w}^{'}}_1^T(k)}& \cdots &{{\bm{w}^{'}}_p^T(k)}
      \end{array}} \right]^T},
      \end{equation}
      \begin{equation}
      {{\bm{w}^{'}}_i}(k) = {\left[ {\begin{array}{*{20}{c}}
      {{w_i^{'}}(k + 1)}& \cdots &{{w_i^{'}}(k + P)}
      \end{array}} \right]^T},
      \end{equation}
      where $i = 1,\cdots,p$. Specifically,
      \begin{equation}
      {\bf{Q}}_i^{'} = \left[ {\begin{array}{*{20}{c}}
        0&{}&{}&{}&{}&{}&{}&{} \\
        {}& \ddots &{}&{}&{}&{}&{}&{} \\
        {}&{}&0&{}&{}&{}&{}&{} \\
        {}&{}&{}&{{s_i}}&{ - {s_i}}&{}&{}&{} \\
        {}&{}&{}&{ - {s_i}}&{{q_i} + 2{s_i}}&{}&{}&{} \\
        {}&{}&{}&{}&{}& \ddots &{}&{} \\
        {}&{}&{}&{}&{}&{}&{{q_i} + 2{s_i}}&{ - {s_i}} \\
        {}&{}&{}&{}&{}&{}&{ - {s_i}}&{{q_i} + {s_i}}
      \end{array}} \right]
      \end{equation}
      and
      $w_i^{'}(k)$ is the solution of the following differential equation:
      \begin{equation}
        \frac{{{s_i}}}{{{q_i}}}\ddot w_i^{'}(k + h) = w_i^{'}(k + h) - {w_i}(k)
        \end{equation}
        where $h = 1,\cdots,P$.
\end{thm}
\begin{proof}
  Appedix A: Proof of Theorem \ref{thm1}\\
  A.1: Derivation of $\bm{Q}_i^{'}$\\
  A.2: Derivation of $w_i^{'}(k)$
\end{proof}
By investigating $\bm{Q}_i^{'}$ and $w_i^{'}(k)$ in the equivalent loss function above, deeper insights about the working principle of the three-term MPC can be gained. Expanding the first term in Eq. (\ref{3to2}) gives
\begin{equation}\label{eq44}
\begin{gathered}
  \left\| {{{\bm{W}}^{'}}(k) - {{\bm{Y}}_P}(k)} \right\|_{{{\bm{Q}}^{'}}}^2 \approx \left\| {{{\bm{W}}^{'}}(k) - {{\bm{Y}}_P}(k)} \right\|_{{\bm{Q}} + 2{\bm{S}}}^2 \hfill \\
  {\quad \quad \quad} - 2\left( {{{\bm{W}}^{'}}(k) - {{\bm{Y}}_P}(k)} \right){\bm{S}}\left( {{{\bm{W}}^{'}}(k + 1) - {{\bm{Y}}_P}(k + 1)} \right) \hfill, \\
\end{gathered}
\end{equation}
where the approximation is due to that the first non-zero value in the diagonal of ${\bf{Q}}_i^{'}$ is $s_i$ and the last diagonal element of ${\bf{Q}}_i^{'}$ is ${q_p} + {s_p}$ instead of ${q_p} + 2{s_p}$. The first term in Eq. (\ref{eq44}) punishes the errors between outputs and their response curves; the second term punishes the oscillation of outputs because opposite-sign errors at adjacent time will make this term positive. The second term in Eq. (\ref{3to2}) is consistent with  the two-term DMC Eq. (\ref{eq11}) to avoid too wild control actions.

To understand the effects of $w_i^{'}(k)$, more explicit expressions about $w_i^{'}(k)$ concerning
two special cases, constant and ramp setpoint, are given in below corollaries.

\begin{cor}\label{cor1}(Constant setpoint).
  Suppose that the setpoint ${{\bf{w}}_i}(k + h) = {{\bf{w}}_i}$ is a constant signal, then
  \begin{equation}\label{eq42}
    w_i^{'}(k + h) = {w_i} + \left( {{w_i} - {y_i}(k)} \right)\left( {1 - {e^{ - h/{\lambda _i}}}} \right)
    \end{equation}
    where $i = 1,\cdots,p$, $h = 1,\cdots,P$, and ${\lambda _i} := \sqrt {{s_i}/{q_i}}$.
\end{cor}

\begin{proof}
  Appendix B: Proof of Corollary \ref{cor1} and \ref{cor2}\\
  B.1: Proof of Corollary 1
\end{proof}

\begin{cor}\label{cor2}(Ramp setpoint).
  Suppose that the setpoint ${{\bf{w}}_i}(k+h) $ is a ramp signal, then
  \begin{equation}\label{wramp}
    w_i^{'}(k + h) = {w_i}(k + h) + \left( {{y_i}(k) - {w_i}(k)} \right){{e^{ - h/{\lambda _i}}}}
    \end{equation}
    where $i = 1,\cdots,p$, $h = 1,\cdots,P$, and ${\lambda _i} := \sqrt {{s_i}/{q_i}}$.
\end{cor}

\begin{proof}
  Appendix B:\\
  B.2: Proof of Corollary \ref{cor2}
\end{proof}

The three-term MPC reshapes the setpoints to new reference trajectories. In the constant setpoint case, they are well-known first-order reference trajectories. This property is very useful for controller tuning. For a multivariable process, after designing ${q_i}$ for each output, the user only needs to enter the expected closed-loop response time for each output, then the output increment weighting matrix $\bm{S}$ can be calculated.
\subsection{Simulation Example}
In this section, simulation results are presented to validate Theorem \ref{thm1}. A 2-input-2-output process with a single sample delay is considered as the process.
\begin{equation}
\begin{array}{l}
{y_1}(k) = \frac{{0.0450{q^{ - 1}} + 0.0450{q^{ - 2}}}}{{1 - 1.7347{q^{ - 1}} + 0.7660{q^{ - 2}}}}{u_1}(k)\\
{\quad \quad \quad \quad \quad \quad \quad \quad \quad} + \frac{{0.1200{q^{ - 1}} + 0.0150{q^{ - 2}}}}{{1 - 1.7347{q^{ - 1}} + 0.7660{q^{ - 2}}}}{u_2}(k)\\
{y_2}(k) = \frac{{0.0700{q^{ - 1}} + 0.0500{q^{ - 2}}}}{{1 - 1.3490{q^{ - 1}} + 0.5140{q^{ - 2}}}}{u_1}(k)\\
{\quad \quad \quad \quad \quad \quad \quad \quad \quad} + \frac{{0.0500{q^{ - 1}} + 0.0200{q^{ - 2}}}}{{1 - 1.3490{q^{ - 1}} + 0.5140{q^{ - 2}}}}{u_2}(k)
\end{array}
\end{equation}
The step response of the process is shown in Fig. \ref{fig: process A};
\begin{figure}
\begin{center}
\includegraphics[width=0.45\textwidth]{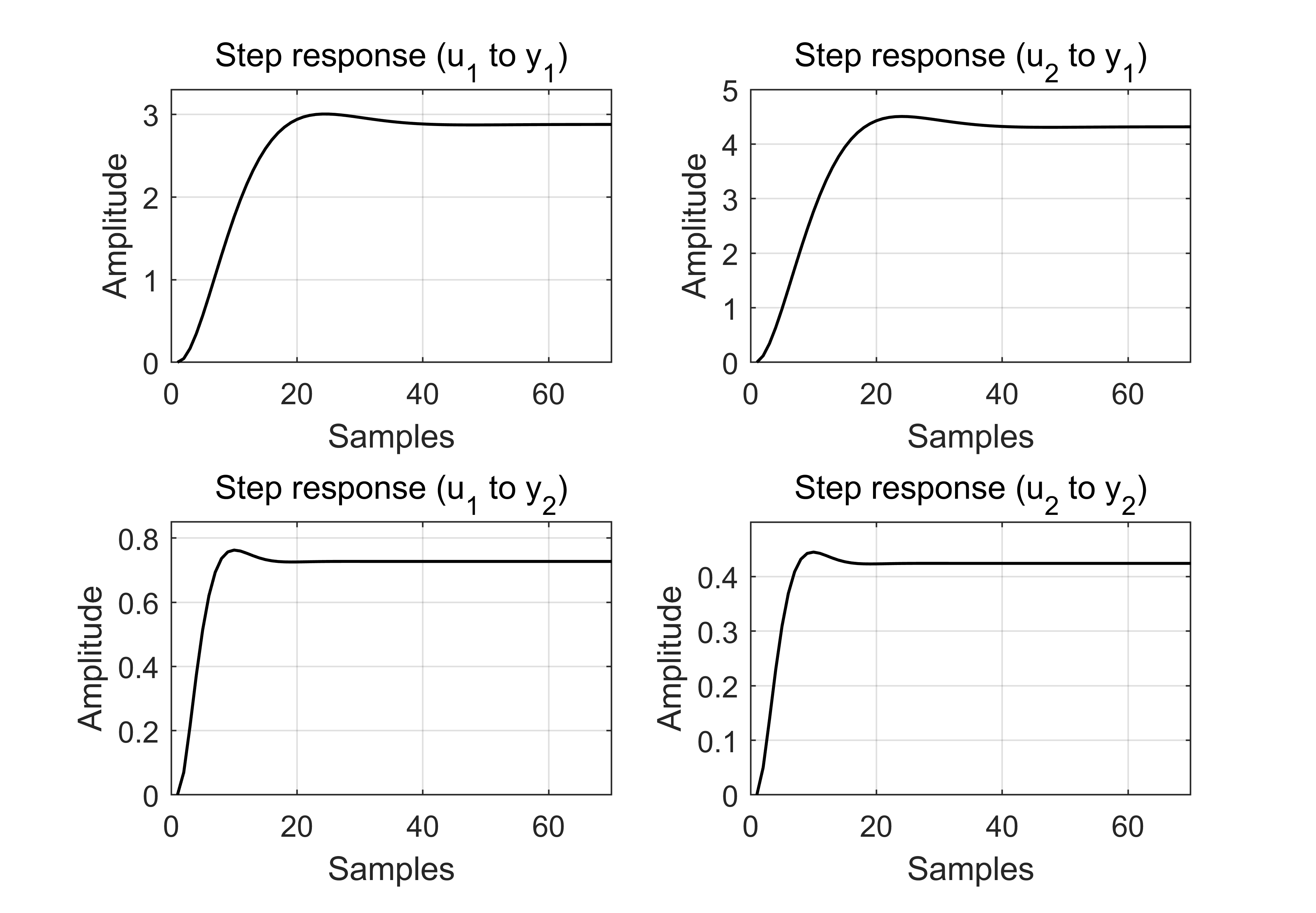}
\caption{Process A}
\label{fig: process A}
\end{center}
\end{figure}

In order to show the relation between the weighting matrices and the equivalent first-order reference curves, a closed-loop step test is performed. The equivalent reference curve $w_i(k)$ is formulated as
\begin{equation}\label{eq58}
{w_i}(k + h) = {r_i}(k) + \left( {{r_i}(k) - {y_i}(k)} \right)\left( {1 - {e^{ - h{T_s}/{\lambda _i}}}} \right)
\end{equation}
where signal $r_i(k)$ denotes the setpoint of the $i^{th}$ output. $h = 1,...,P$. $\lambda_i$ denotes the time constant of $i^{th}$ reference curve. The horizon settings are: horizon of dynamics $N=60$, prediction horizon $P=40$, control horizon $M=10$. The settings of weighting parameters are as follows

$q_1=q_2=50$, $r_1=r_1=1$, $s_1=800$, $s_2=200$;

The parameter of equivalent first order reference curves are as follows

$\lambda_1=4$, $\lambda_2=2$

Perform the closed-loop step test and the result is shown in Fig. \ref{fig4} and Fig. \ref{fig5}. One can see that the obtained closed-loop step responses match the desired reference curve very closely. Note that $y_1$ fits its reference curve better than $y_2$ does because the weighting of $y_1$ is much stronger than that of $y_2$ in the two-term QP due to larger $s_1$ and larger model gains.

\begin{figure}
\begin{center}
\includegraphics[width=0.45\textwidth]{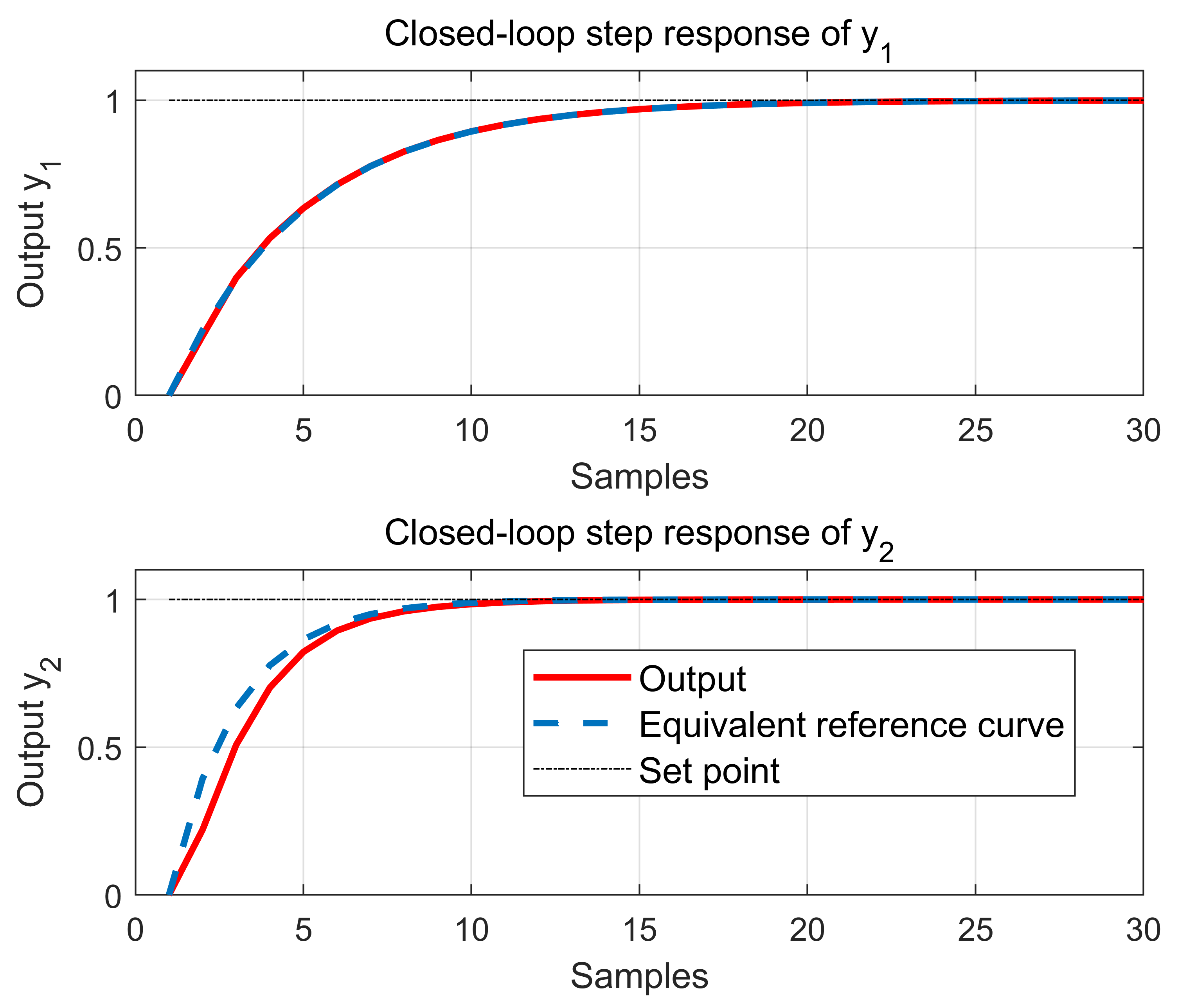}
\caption{Outputs of the process A, step response test}
\label{fig4}
\end{center}
\end{figure}

\begin{figure}
\begin{center}
\includegraphics[width=0.45\textwidth]{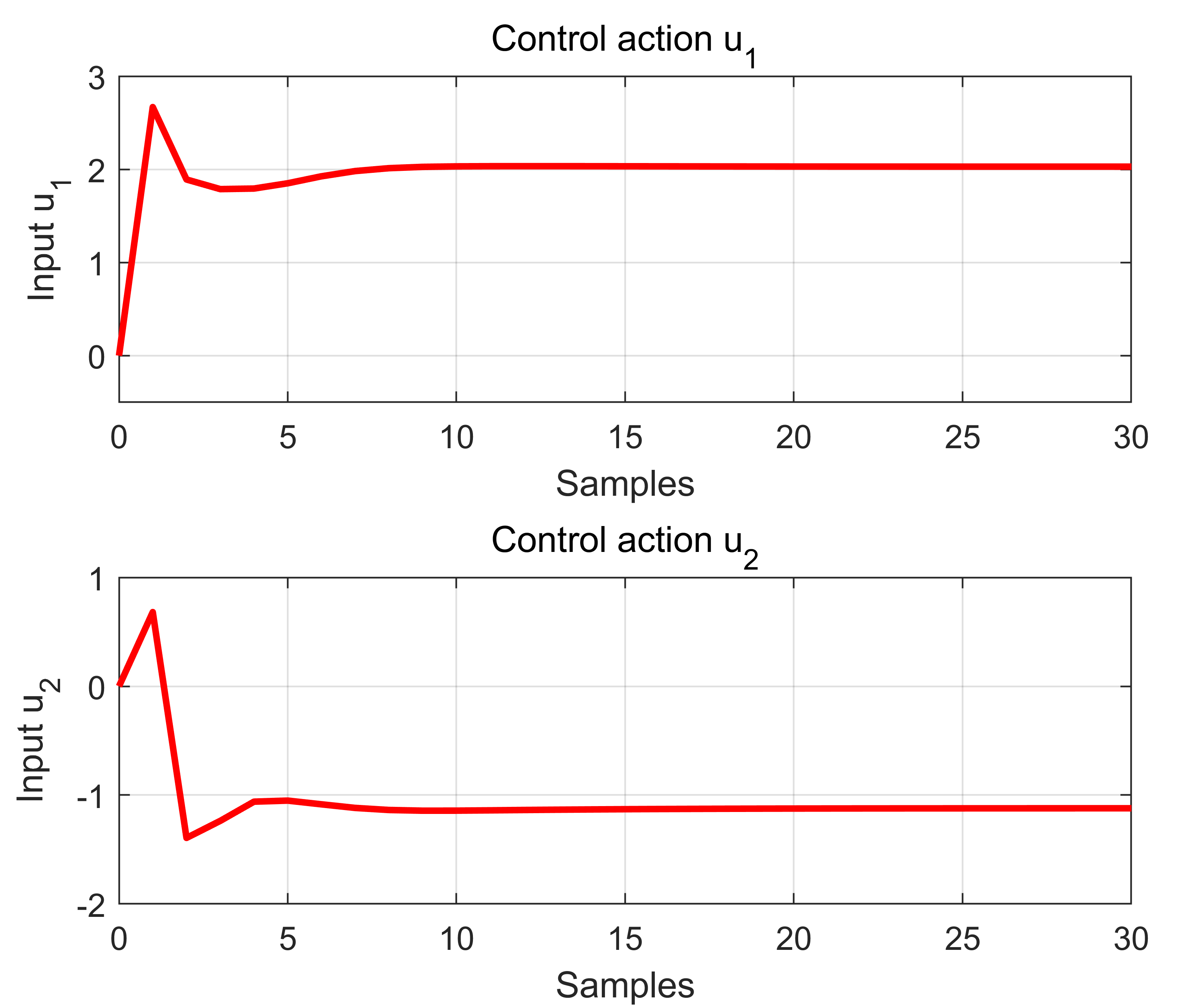}
\caption{Inputs of the process A, step response test}
\label{fig5}
\end{center}
\end{figure}

\subsection{Formula of ideal closed-loop response curves}
In this section, the ideal closed-loop response of a multivariable system controlled by an unconstrained three-term DMC will be derived supposing the weighing matrix $\bm{R} = \bm{0}$. Some assumptions are made before the derivation:
\begin{ass}\label{ass1}
The offsets between references and outputs of the disturbance-free system controlled by the three-term DMC are constant after $P$ steps, i.e.,
\begin{equation}
  w_i^{'}(k+h) - y_i(k+h) = \mathrm{constant}
\end{equation}
for $i = 1,\cdots,p$, $h\geq P$, $h\in\mathbb{N}^+$.
\end{ass}

\begin{ass}\label{ass2}
The model horizon $N$ and prediction horizon $P$ are large enough such that the predicted output $\bm{y}_P(k)$ equals the actual output $\bm{y}(k)$ of the disturbance-free system.
\end{ass}

When $\bm{R} = \bm{0}$, it is not difficult to realize Assumption~\ref{ass1} because there is no restriction for the control action. However, there may exist nontrivial systems which has been excluded by Assumption~\ref{ass1}. Notice also that Assumption~\ref{ass1} has implied the closed-loop stability. Assumption~\ref{ass2} is posed to derive the exact expression of the closed-loop response. Now the following theorem can be proved:

\begin{thm}\label{thm2}
  When $\bm{R} = \bm{0}$ and under Assumptions \ref{ass1} and \ref{ass2}, the closed-loop response of a disturbance-free system controlled by the unconstrained three-term DMC is given by:
  \begin{enumerate}
    \item $d_i = 1$:\begin{equation}\label{est_single}
      y_i(k + h) = {w_i^{'}}(k + h)
      \end{equation}
    \item $d_i > 1$: \begin{equation}\label{est_multi}
      {y_i}(k + h) = \left\{ {\begin{array}{*{20}{c}}
      {0{\quad \quad \quad \quad \quad \quad},h \le {d_i} - 1}\\
      {{w_i^{'}}(k + h) - {w_i^{'}}(k + {d_i} - 1)\alpha _i^{h - d_i + 1}{\rm{   }},h > {d_i} - 1}
      \end{array}} \right.
      \end{equation}
      \begin{equation}
      {\alpha_i} = \frac{{\left( {\frac{q_i}{s_i} + 2} \right) - \sqrt {{{\left( {\frac{q_i}{s_i} + 2} \right)}^2} - 4} }}{2}
      \end{equation}
\end{enumerate}
where $i = 1, \cdots p$, $h = 1, \cdots P$, and $d_i$ is the output delay at $i$th output.
\end{thm}

\begin{proof}
  Appendix C: Proof of Theorem \ref{thm2}.
  \end{proof}

Theorem \ref{thm2} gives the formula for calculating the closed-loop responses when the input control move is not punished ($\bm{R}=\bm{0}$). It considers time delays so that the tuning is more accurate than the reference curves which cannot consider time delays. The calculated closed-loop response can be used in the tuning of the three-term DMC.
\subsection{Simulation Example}
Consider a 2-input-2-output system as follows:
\begin{equation}
\begin{array}{l}
{y_1}(k) = \frac{{0.0450{q^{ - 1}} + 0.0450{q^{ - 2}}}}{{1 - 1.7347{q^{ - 1}} + 0.7660{q^{ - 2}}}}{u_1}(k - 1)\\
{\quad \quad \quad \quad \quad \quad \quad \quad \quad} + \frac{{0.1200{q^{ - 1}} + 0.0150{q^{ - 2}}}}{{1 - 1.7347{q^{ - 1}} + 0.7660{q^{ - 2}}}}{u_2}(k - 1)\\
{y_2}(k) = \frac{{0.0700{q^{ - 1}} + 0.0500{q^{ - 2}}}}{{1 - 1.3490{q^{ - 1}} + 0.5140{q^{ - 2}}}}{u_1}(k - 4)\\
{\quad \quad \quad \quad \quad \quad \quad \quad \quad} + \frac{{0.0500{q^{ - 1}} + 0.0200{q^{ - 2}}}}{{1 - 1.3490{q^{ - 1}} + 0.5140{q^{ - 2}}}}{u_2}(k - 4)
\end{array}
\end{equation}
The step response of the process is illustrated in Fig. \ref{fig: process B}, with a delay of two samples for the first output and a delay of four samples for the second output. The three-term DMC is used with the following parameter settings:

\begin{enumerate}
\item Horizon of dynamic $N=55$, prediction horizon $P=45$, and control horizon $M=10$;
\item Weighting factors: $q_1=1$, $s_1=1$, $r_1=0.0001$; $q_2=1$, $s_2=2$, $r_2=0.0001$.
\end{enumerate}

The settings ensure that $\frac{{{s_1}}}{{{q_1}}} = 1$, $\frac{{{s_2}}}{{{q_2}}} = 2$, and that $r_j$ is assigned a very small value to prevent numerical issues.

First, the closed-loop step responses are shown. According to Theorem \ref{thm1} and Corollary \ref{cor1}, the prediction of the output is given by
\begin{equation}
{y_i}(k + h) = \left\{ {\begin{array}{*{20}{c}}
{0{\quad \quad \quad \quad \quad \quad},h \le {d_i} - 1}\\
{{w_i^{'}}(k + h) - {w_i^{'}}(k + {d_i} - 1)\alpha _i^{h - d_i + 1}{\rm{   }},h > {d_i} - 1}
\end{array}} \right.
\end{equation}
Here, the equivalent reference curve $w_i^{'}(k)$ is given by
\begin{equation}
w_i^{'}(k + h) = {r_i}(k) + \left( {{r_i}(k) - {y_i}(k)} \right)\left( {1 - {e^{ - h{T_s}/{\lambda _i}}}} \right)
\end{equation}
where ${\lambda _1} = 1$, ${\lambda _2} = \sqrt 2 $, ${\alpha _1} = 0.382$, and ${\alpha _2} = 0.5$. The set point sequence ${r_i}(k) = 1$. The result is shown in Fig. \ref{steppre}. The prediction of the closed-loop step response is accurate.

Second, the closed-loop ramp responses are shown. According to Theorem \ref{thm1} and Corollary \ref{cor2}, the prediction of the output is given by
\begin{equation}
{y_i}(k + h) = \left\{ {\begin{array}{*{20}{c}}
{0{\quad \quad \quad \quad \quad \quad},h \le {d_i} - 1}\\
{{w_i}(k + h) - {w_i}(k + {d_i} - 1)\alpha _i^{h - d_i + 1}{\rm{   }},h > {d_i} - 1}
\end{array}} \right.
\end{equation}
Here, $w_i(k)$ is the setpoint sequence and ${\alpha _1} = 0.382$, ${\alpha _2} = 0.5$. The result is shown in Fig. \ref{ramppre}. Again, the prediction of the closed-loop ramp response is accurate.

The simulation shows that the closed-loop response formula

\begin{figure}
    \begin{center}
    \includegraphics[width=0.45\textwidth]{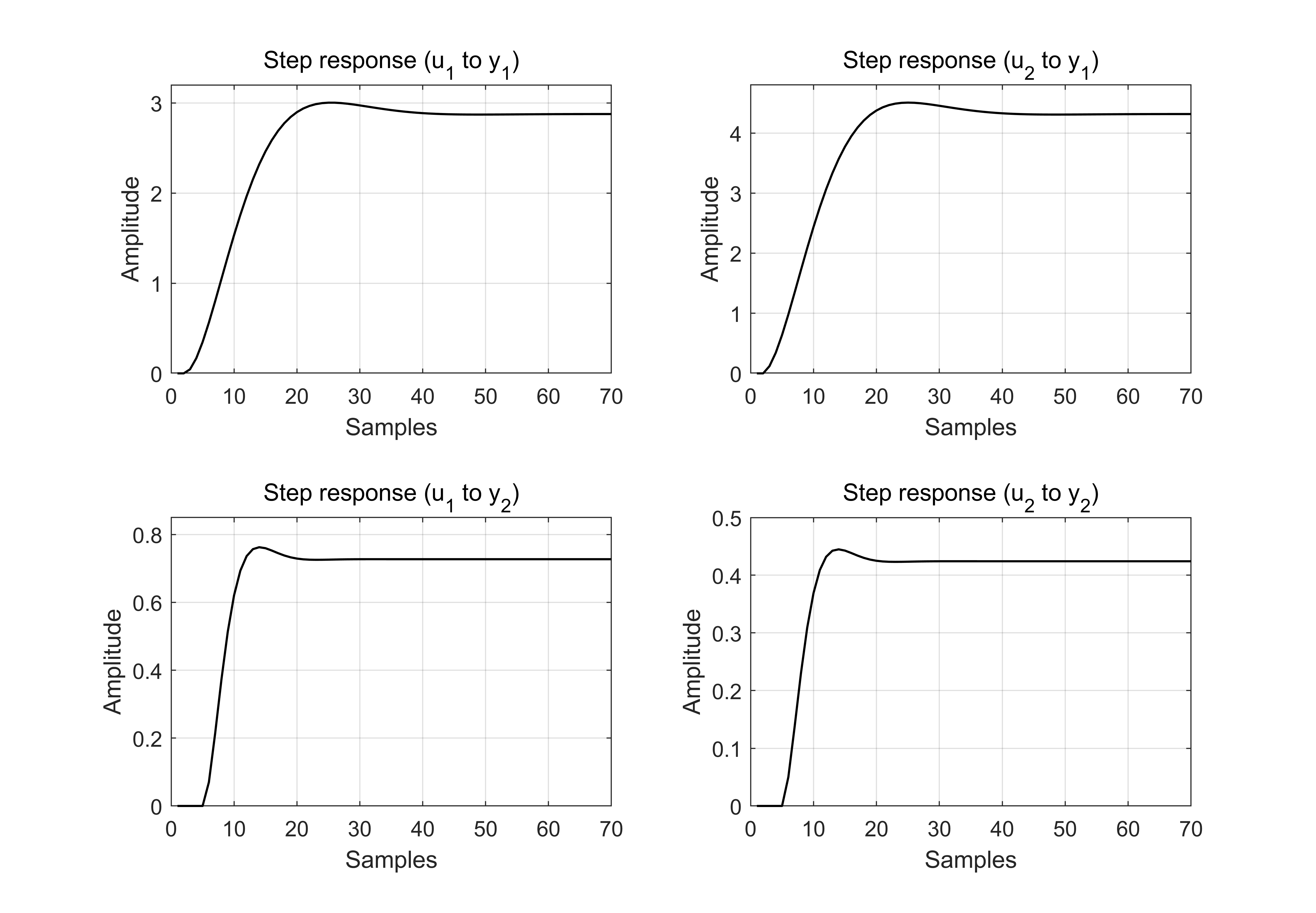}
    \caption{Process B}
    \label{fig: process B}
    \end{center}
\end{figure}
\begin{figure}
    \begin{center}
    \includegraphics[width=0.45\textwidth]{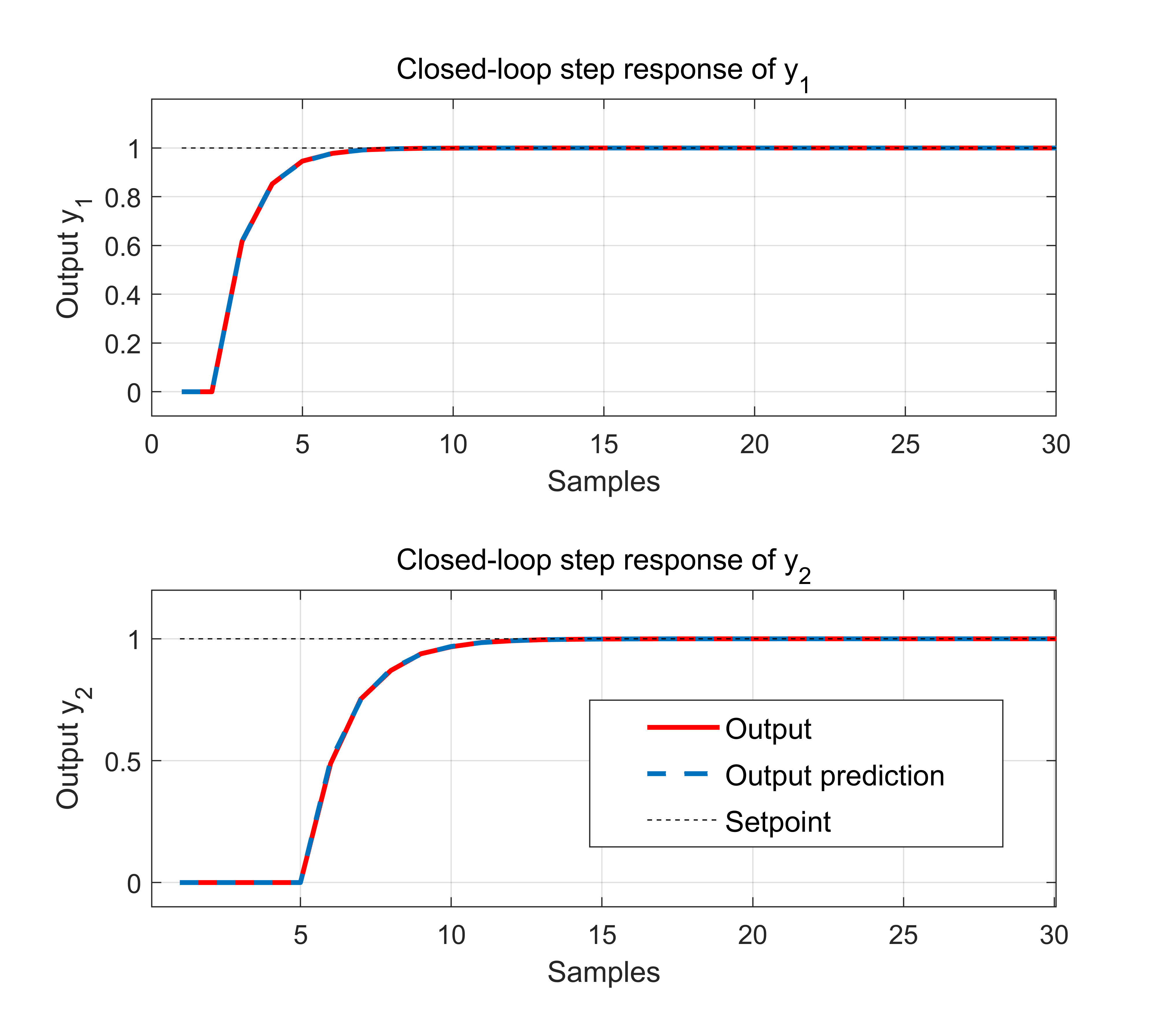}
    \caption{Closed-loop step response of process B}
    \label{steppre}
    \end{center}
\end{figure}

\begin{figure}
    \begin{center}
    \includegraphics[width=0.45\textwidth]{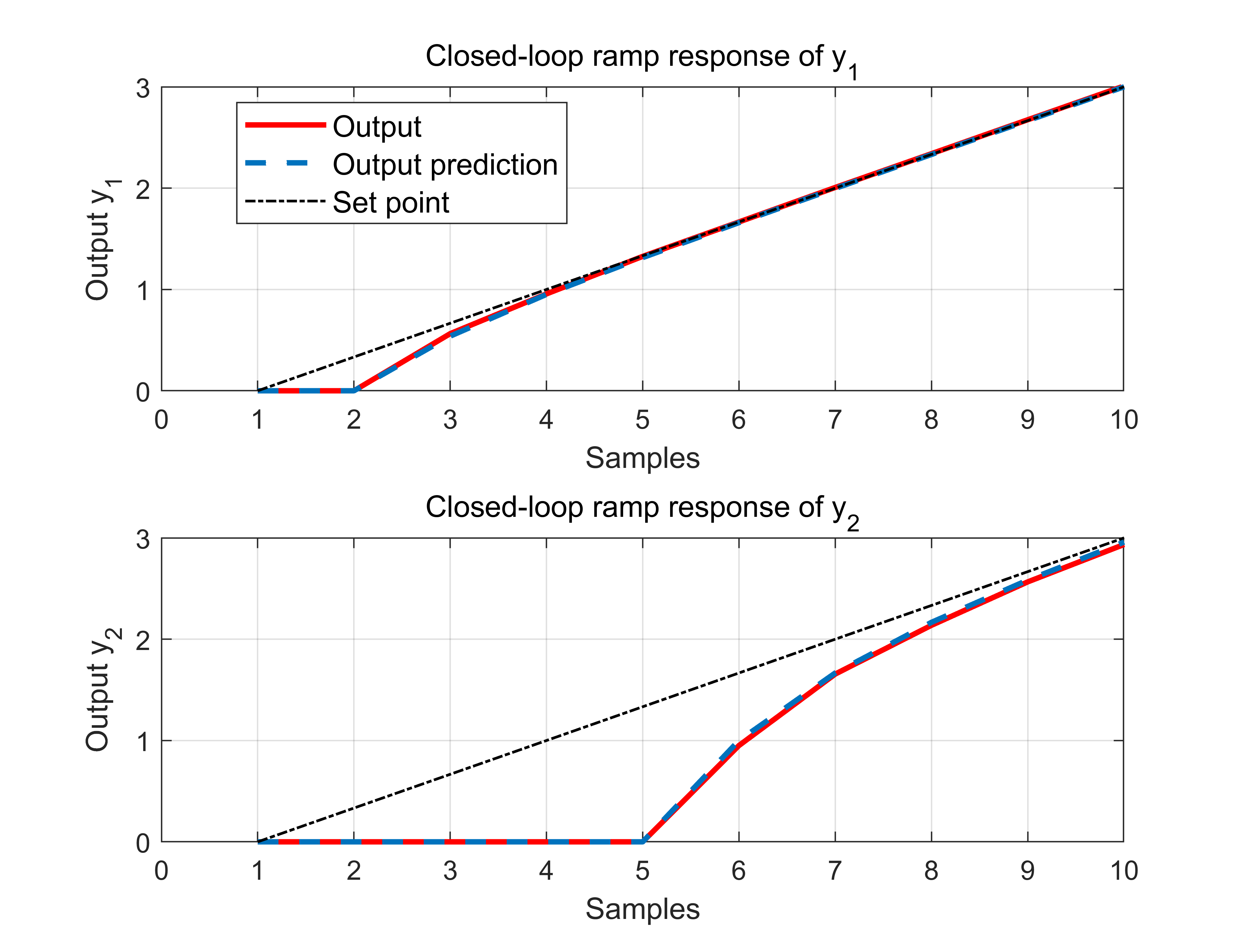}
    \caption{Closed-loop ramp response of process B}
    \label{ramppre}
    \end{center}
\end{figure}

\section{Tuning of the three-term DMC}
In industrial MPC applications, not only high performance but also user-friendliness in controller tuning is highly desired. In large-scale industrial MPC projects, how to determine the weighting matrixes $\bm{Q}$, $\bm{R}$, and  $\bm{S}$ is not straightforward, and optimizing them is even more difficult. A good way to design a controller is to let the user determine the desired closed-loop response. This is called loop-shaping for frequency domain design methods (~\cite{Mcfarlane1992}). In this work, first, a time domain loop-shaping is developed, where the user determines desired closed-loop step responses for process outputs. Then, a tuning method for disturbance reduction is developed.

\subsection{The prediction error method of system identification}\label{si}
To facilitate tuning and control in model predictive control (MPC), a process model is required. The three-term dynamic matrix control (DMC) tuning method utilizes the model to obtain closed-loop responses and perform simulation-based tuning. To achieve this, the prediction error method of system identification is first introduced.

Consider a linear, time-invariant process with $m$ inputs and $p$ outputs. The process's autoregressive moving average with exogenous input (ARMAX) model is expressed as follows:
\begin{equation}\label{system}
{\bm{y}}(k) = {{\bm{G}}_0}({q^{ - 1}}){\bm{u}}(k) + {{\bm{H}}_0}({q^{ - 1}}){\bm{e}}(k)
\end{equation}
with
\begin{equation}\label{mfd1}
{{\bm{G}}_0}({q^{ - 1}}) = {\bm{A}}_0^{ - 1}({q^{ - 1}}){{\bm{B}}_0}({q^{ - 1}})
\end{equation}
\begin{equation}\label{mfd2}
{{\bm{H}}_0}({q^{ - 1}}) = {\bm{A}}_0^{ - 1}({q^{ - 1}}){{\bm{C}}_0}({q^{ - 1}})
\end{equation}
where
\begin{equation}
{{\bm{A}}_0}({q^{ - 1}}) = diag\{ A_0^1({q^{ - 1}}), \cdots ,A_0^p({q^{ - 1}})\}
\end{equation}

\begin{equation}
{{\bm{B}}_0}({q^{ - 1}}) = \left[ {\begin{array}{*{20}{c}}
{B_0^{11}({q^{ - 1}})}& \cdots &{B_0^{1m}({q^{ - 1}})}\\
 \vdots & \ddots & \vdots \\
{B_0^{p1}({q^{ - 1}})}& \cdots &{B_0^{pm}({q^{ - 1}})}
\end{array}} \right]
\end{equation}

\begin{equation}
{{\bm{C}}_0}({q^{ - 1}}) = diag\{ C_0^1({q^{ - 1}}), \cdots ,C_0^p({q^{ - 1}})\}
\end{equation}
where  ${{\bm{A}}_0}({q^{ - 1}})$ and ${{\bm{C}}_0}({q^{ - 1}})$ are $p \times p$ diagonal matrices of polynomials with unit leading coefficients, ${{\bm{B}}_0}({q^{ - 1}})$ is a $p \times m$ matrix of polynomials. Denote ${\bm{y}}(k)$ and ${\bm{u}}(k)$ as the output vector and the input vector at sampling time $k$ and denote the signal ${\bm{e}}(k)$ as the white noise vector.

\begin{equation}
{\bm{y}}(k) = {\left[ {\begin{array}{*{20}{c}}
{{y_1}(k)}&{...}&{{y_p}(k)}
\end{array}} \right]^T}
\end{equation}
\begin{equation}
{\bm{u}}(k) = {\left[ {\begin{array}{*{20}{c}}
{{u_1}(k)}&{...}&{{u_m}(k)}
\end{array}} \right]^T}
\end{equation}
\begin{equation}
{\bm{e}}(k) = {\left[ {\begin{array}{*{20}{c}}
{{e_1}(k)}&{...}&{{e_p}(k)}
\end{array}} \right]^T}
\end{equation}
where $y_i(k)$ and $u_j(k)$ are the $i$th output and the $j^{th}$ input. The signal $e_i(k)$ is the white noise of $i^{th}$ channel with variance $\sigma _{{e_i}}^2 > 0$.
When  ${{\bm{A}}_0}({q^{ - 1}})$ and ${{\bm{C}}_0}({q^{ - 1}})$ are diagonal matrices, the ARMAX model Eq. (\ref{mfd1}-\ref{mfd2}) is  called diagonal form matrix fraction description of the process Eq. (\ref{system}). This is the simplest form for describing multi-input multi-output (MIMO) processes. In this form, the MIMO process Eq. (\ref{system}) is decomposed into $p$ subsystems,
\begin{equation}
{y_i}(k) = \frac{1}{{A_0^i({q^{ - 1}})}}\sum\limits_{j = 1}^m {B_0^{ij}({q^{ - 1}}){u_j}(k)}  + \frac{{C_0^i({q^{ - 1}})}}{{A_0^i({q^{ - 1}})}}{e_i}(k)
\end{equation}
For each output, the one-step-ahead prediction error is defined as
\begin{equation}
{\varepsilon _i}(k) = \frac{{{A^i}({q^{ - 1}})}}{{{C^i}({q^{ - 1}})}}\left[ {{y_i}(k) - \sum\limits_{j = 1}^m {\frac{{{B^{ij}}({q^{ - 1}})}}{{{A^i}({q^{ - 1}})}}{u_j}(k)} } \right]
\end{equation}
where ${{A^i}({q^{ - 1}})}$, ${{B^{ij}}({q^{ - 1}})}$, ${{C^i}({q^{ - 1}})}$ are polynomials.

Let the model parameters of each MISO subsystem be collected in ${\theta ^i}$. Using a set of input-output data of length $N$, the estimate $\hat \theta _N^i$ is calculated by minimizing the prediction error loss function:
\begin{equation}
\hat \theta _N^i = \arg \mathop {\min }\limits_\theta  \frac{1}{N}\sum\limits_{k = 1}^N {\varepsilon _i^T(k|\theta ){\varepsilon _i}(k|\theta )}
\end{equation}
Then, one obtains the model of the process
\begin{equation}\label{getmodel}
{{\hat y}_i}(k) = \frac{1}{{{A^i}({q^{ - 1}})}}\sum\limits_{j = 1}^m {{B^{ij}}({q^{ - 1}}){u_j}(k)}
\end{equation}
The unmeasured disturbance $v_i(k)$ can be estimated as
\begin{equation}\label{getv}
{v_i}(k) = {y_i}(k) - {{\hat y}_i}(k)
\end{equation}
Eq. (\ref{getmodel}) and Eq. (\ref{getv}) can be used in the simulation-based tuning procedure.
\subsection{Tuning for closed-loop step responses}
Similar to frequency domain loop-shaping, one can use time domain loop-shaping by setting closed-loop step responses for each output. In Theorem \ref{thm2}, one knows how to determine the ratios of weighting matrices $\bf{S}$ and $\bf{Q}$ in order to obtain desired closed-loop settling times of outputs. The remaining question is how to determine the weighting matrices $\bf{Q}$ and $\bf{R}$.

Before the tuning procedure, the input and output variables are normalized
\begin{equation}\label{eq45}
{q_i} = \frac{{{k_y}}}{{{{\left( {{y_{i,\max }} - {y_{i,\min }}} \right)}^2}}}
\end{equation}
\begin{equation}\label{eq46}
{r_j} = \frac{{{k_u}}}{{{{\left( {{u_{j,\max }} - {u_{j,\min }}} \right)}^2}}}
\end{equation}
where ${{y_{i,\max }}}$ and ${{y_{i,\min }}}$ denote the upper limit and the lower limit of the $i^{th}$ output, and ${{u_{j,\max }}}$ and ${{u_{j,\min }}}$ denote the upper limit and the lower limit of the $j^{th}$ input. ${{k_y}}$ and ${{k_u}}$ are scalars. Let ${k_{yu}}={k_y}/{k_u}$, which determines the closed-loop bandwidth.

In general, fixing the input weighting matrix $\bf{R}$ and increasing the output weighting matrix $\bf{Q}$ will increase the tracking performance, but will reduce the robustness to model uncertainty and may cause overshoot and oscillation. In industrial applications, it is desirable that:

\begin{enumerate}
    \item The outputs follow their set point trajectories closely;
    \item The control actions do not move too wildly.
\end{enumerate}

The second requirement is for smooth control actions and robustness of the controlled system. Observations from many industrial applications and simulations show that the control action wildness (or smoothness) are related to the overshoots of input movements during a closed-loop step test where all set points are changed by one unit simultaneously. The sizes of input overshoots in a closed-loop step test can be used to measure the control action wildness: the larger the overshoots, the wilder the control action and lower the robustness. The user can set upper bounds for input overshoots when tuning the controller.

The tuning for the closed-loop step response is summarized as follow:

%
%
%
%
%
\begin{table}[]
\caption{Tuning procedure of closed-loop step response}
\begin{center}
\begin{tabular}{|c|l|}
\hline
\multicolumn{1}{|l|}{} & \multicolumn{1}{c|}{Procedure}                                                                                                                                                                                                                          \\ \hline
Step 1 & \begin{tabular}[c]{@{}l@{}}Normalize the input and output variables \\ according to Eq. (\ref{eq45}) and Eq. (\ref{eq46});\end{tabular}             \\ \hline
Step 2 & \begin{tabular}[c]{@{}l@{}}Specify closed-loop time constants for \\ each output and upper bounds of input \\ control action overshot;\end{tabular}                                 \\ \hline
Step 3 & \begin{tabular}[c]{@{}l@{}}Determine the ratio of $\bf{S}$ and $\bf{Q}$ according \\ to closed-loop step response estimation \\ (Theorem \ref{thm2});\end{tabular} \\ \hline
Step 4 & \begin{tabular}[c]{@{}l@{}}Start closed-loop step test simulation from \\ a small ratio $k_{yu}$;\end{tabular}                                                                      \\ \hline
Step 5 & \begin{tabular}[c]{@{}l@{}}Increase the $k_{yu}$ in closed-loop step simu-\\ lations until some control action over-shoots \\ reach their overshoot upper bounds.\end{tabular}      \\ \hline
Step 6 & \begin{tabular}[c]{@{}l@{}}Fine tune the elements $q_i$ in closed-\\ loop step test until most control action \\ overshoots reach their bounds.\end{tabular}                        \\ \hline
\end{tabular}
\end{center}
\label{table1}
\end{table}

For closed-loop response times, set:
\begin{enumerate}
    \item slow controller, set them around 0.8 times their open loop response times;
    \item medium fast controller, set them around 0.4 times the open loop response times;
    \item fast controller, set them around 0.2 times the open loop response times.
\end{enumerate}

Note that time delays are not included in the above calculation.

For upper bounds of control action overshoots, set:
\begin{enumerate}
    \item no overshoot for smooth control;
    \item 50\% overshoot for medium strong control action;
    \item 100\% overshoot for strong control action.
\end{enumerate}

The two settings are related. For example, reducing the closed-loop response times can increase control action overshoots. One needs to consider model accuracy when setting the two sets of parameters. If the process model is very accurate, short closed-loop response times and large control action overshoots could be used; if the model is not accurate, larger response times and smaller overshoots should be chosen.

The time-domain loop shaping tuning is by no means complete that can cover all situations, but it is simple, intuitive and effective for process industry applications. The parameter settings given above are not very strict in applications and the user can adjust them in their situations.

\subsection{Tuning for optimal disturbance reduction}
A more important control performance in process industries is disturbance reduction rather than setpoint tracking. The time-domain loop shaping tuning can be used for disturbance reduction in a indirect manner. However, a tuning method that directly optimizes disturbance reduction could be more desirable.

Optimal disturbance (noise) reduction means that, for a given class of unmeasured disturbances, the controller parameters are tuned to minimize output variances under some constraints in robustness. As mentioned in subsection \ref{si}, unmeasured disturbances at process outputs can be estimated in the system identification as in Eq. (\ref{getv}). Here again, the upper bounds of input control action overshoots in closed-loop step tests are used as constraints.

Before the tuning procedure, the input and output variables are normalized like Eq. (\ref{eq45}) and Eq. (\ref{eq46}). Then the tuning procedure can be formulated as an optimization problem:
\begin{equation}\label{optimization}
\begin{array}{l}
\mathop {\min }\limits_{{k_{yu}}} {I_\sigma } = \frac{1}{p}\sum\limits_{i = 1}^p {\frac{{{\sigma _{{y_i}}}}}{{\left| {{y_{i,\max }} - {y_{i,\min }}} \right|}}} \\
s.t.\left| {u_j^{os}} \right| \le \left| {u_j^{bound}} \right|
\end{array}
\end{equation}
In the given expression, ${I_\sigma}$ represents the standard deviation of the $i^{th}$ output, while ${u_j^{os}}$ denotes the overshoot of the control action $u_j$ in the closed-loop step test. Moreover, ${u_j^{bound}}$ represents the upper bound of the overshoot for the $j^{th}$ input in the closed-loop step test. It is important to note that each input is associated with a specific bound.

Based on the industrial experience of the authors, the following tuning method for optimal disturbance reduction is proposed.

%
%
%
%
%
%
%
\begin{table}[]
\caption{Tuning procedure of disturbance reduction}
\begin{center}
\begin{tabular}{|c|l|}
\hline
\multicolumn{1}{|l|}{} & \multicolumn{1}{c|}{Procedure}                                                                                                                                                                                                                          \\ \hline
Step 1                 & \begin{tabular}[c]{@{}l@{}}Normalize the input and output variables \\ according to Eq. (\ref{eq45}) and Eq. (\ref{eq46});\end{tabular}                                                                                    \\ \hline
Step 2                 & \begin{tabular}[c]{@{}l@{}}Set closed-loop step response times for \\ each output;\end{tabular}                                                                                                                                                      \\ \hline
Step 3                 & \begin{tabular}[c]{@{}l@{}}Set upper bounds of input control action \\ overshoots in closed-loop step test;\end{tabular}                                                                                                                             \\ \hline
Step 4                 & \begin{tabular}[c]{@{}l@{}}Determine the ratio of $\bf{S}$ and $\bf{Q}$ according \\ to closed-loop step response estimation \\ (Theorem \ref{thm2});\end{tabular}                                                                  \\ \hline
Step 5                 & \begin{tabular}[c]{@{}l@{}}Run the simulation with increasing $k_{yu}$ \\ until some inputs reach their overshoot \\ upper bounds in the closed-loop step test;\end{tabular}                                                                         \\ \hline
Step 6                 & \begin{tabular}[c]{@{}l@{}}Run the simulation of disturbance reduction \\ with the same $k_{yu}$ as in step 4. Calculate \\ the performance index ${I_\sigma }$;\end{tabular}                                                          \\ \hline
Step 7                 & \begin{tabular}[c]{@{}l@{}}Plot the relation curve of ${I_\sigma }$ and $k_{yu}$; Select \\ $k_{yu}$ which leads to minimum value of ${I_\sigma }$. \\ Calculate the weighting factor: $q_i$, $r_i$, $s_i$;\end{tabular} \\ \hline
Step 8                 & \begin{tabular}[c]{@{}l@{}}Fine-tune the elements $q_i$ in the closed-loop \\ step test until most control action overshoots \\ reach their bounds.\end{tabular}                                                                                     \\ \hline
\end{tabular}
\end{center}
\label{table2}
\end{table}

\begin{mark_}
The tuning for optimal disturbance reduction differs from the tuning for closed-loop step response in two aspects: (1) In the tuning for optimal disturbance reduction, a small ratio of $\frac{{{s_i}}}{{{q_i}}}$ will be given in order to react on disturbance quickly; in the tuning for step response, larger ratios of $\frac{{{s_i}}}{{{q_i}}}$ might be given for high robustness; (2) In the tuning for optimal disturbance reduction, the estimates of unmeasured disturbances are used in simulations; in the tuning for step response, the unmeasured disturbances are not considered.
\end{mark_}

\section{Performance and robustness comparison}
In this section, the comparison of the three-term DMC and the two-term DMC will be carried out both theoretically and through simulations, focusing on control performance and robustness.
\subsection{An approach to comparing two control methods}
To provide a more tangible demonstration of the advantages of the three-term DMC over the two-term DMC, a novel method for comparing controller performance and robustness will be proposed. The method involves designing controllers using Method A (referring to the three-term DMC) and controllers using Method B (referring to the two-term DMC). Using a given process model, closed-loop simulations will be performed to evaluate the controllers' performance and robustness. Specifically, step response and disturbance reduction simulations will be conducted while varying the controllers' parameters for both methods. The performance and robustness of each method will be evaluated to determine which one performs better. To facilitate the comparison, control error and control action will be defined and measured in each simulation.

The closed-loop control system is shown in Fig. \ref{fig: thm3}. Let $G$ and $\hat G$ denote the true plant and the plant model, respectively, with $\Delta G = G - \hat G$ representing the additive model error. $C$ represents the equivalent linear controller obtained via MPC solution Eq. (\ref{eq15}) or Eq. (\ref{eq29}), $w(k)$ denotes the vector of setpoints, and $u(k)$ and $y(k)$ denote the input and output vectors of the plant, respectively. The unmeasured disturbance is denoted by $v(k)$. To evaluate the abilities of the two-term DMC and the three-term DMC in disturbance rejection and robustness, one can adopt the standard deviation of input and output signals, as introduced in the previous section. One can compare the performance of these control algorithms under the same disturbance scenario.
\begin{figure}
\begin{center}
\includegraphics[width=0.45\textwidth]{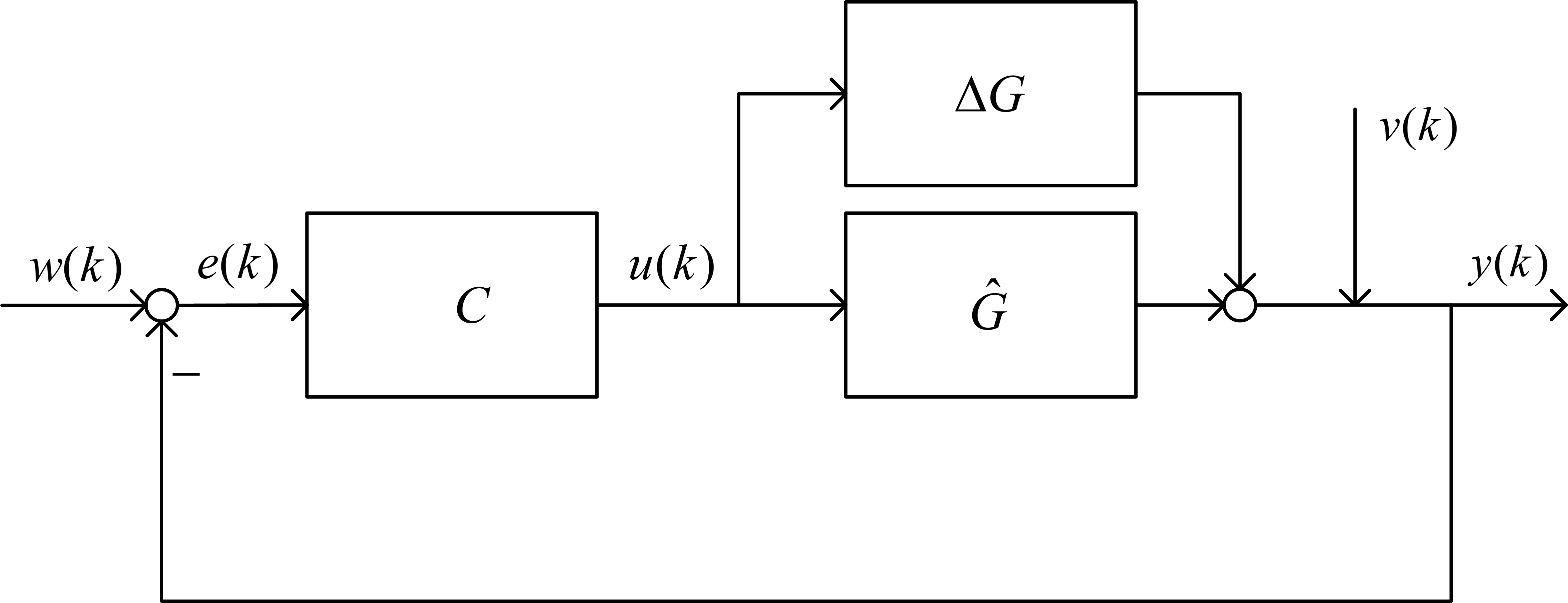}
\caption{Uncertainty description involving additive perturbations}
\label{fig: thm3}
\end{center}
\end{figure}
\begin{definition}
(Control error and control action). Perform a closed-loop control simulation with zero initial condition, the control error is a norm of $w(k)-y(k)$ during the simulation; the control action is a norm of control signal $u(k)$ during the simulation.
\end{definition}
Typical simulations are step response simulation, and disturbance reduction simulation using estimates of unmeasured disturbances obtained from system identification. Typical norm used is the $L_2$ norm.

In the literature, some researchers have compared the control performance of different methods solely based on their control errors, which can be inadequate for an accurate assessment. This is because it is possible for a controller using Method A to have a smaller control error than a controller using Method B, but also have a larger control action. In such a case, Method B may be tuned faster to achieve a smaller control error with a smaller control action than Method A. To address this issue, we propose a procedure for comparing the control performance and robustness of different methods, which considers both control errors and control actions.

\begin{definition}
(Control method comparison). Given control Method A and Method B. Perform closed-loop step response simulations or disturbance reduction simulations with varying parameters. Then Method A is said to be higher performing and more robust against model errors than Method B if

(1) (in terms of performance) using the same control action, the controller of Method A has a smaller control error than Controller B has;

(2) or, equivalently, (in terms of robustness) for the same control error, the controller of Method A uses smaller control action than the controller of Method B does;

(3) or, equivalently, (in terms of performance and robustness) with smaller control error, the controller of Method A uses smaller control action than the controller of Method B does.

\end{definition}

This definition considers both control error and control action, making it possible to understand the performance and robustness comprehensively. It is easy to understand that a smaller control error indicates higher performance. As shown in Fig. \ref{fig: thm3}, a smaller control action implies that the norm of the closed-loop transfer function matrix from the setpoint $w(t)$ or disturbance $v(t)$ is smaller, which, according robust control theory, will make the closed-loop system more robustly stable against additive model errors, see, e.g., \cite{Morari1989}. Assume that the closed-loop system is not close to the boundary of stability. Then, it is not difficult to show that the three comparing options are equivalent.

Consider the closed-loop system depicted in Fig. \ref{fig: thm3}, and find the optimal stabilizing controllers using two distinct methods, which minimize the control errors given by:
\begin{equation}
{J_e} = \sum\limits_{i = 1}^p {\frac{{{\sigma _{{y_i}}}}}{{\left| {{y_{i,\max }} - {y_{i,\min }}} \right|}}}
\end{equation}
subject to the constraint that the control actions of both controllers are equal, expressed as:
\begin{equation}
{J_u} = \sum\limits_{j = 1}^m {\frac{{{\sigma _{{u_j}}}}}{{\left| {{u_{j,\max }} - {u_{j,\min }}} \right|}}}
\end{equation}
Here, $\sigma_{y_i}$ and $\sigma_{u_j}$ denote the standard deviation of the $i^{th}$ output and $j^{th}$ input, respectively. The performance index for comparing the two controllers concerning disturbance reduction is defined in option (1) of \textbf{Definition 2} in the control method comparison as follows:
\begin{equation}\label{eq33}
J_w = \sum_{i=1}^p \frac{\sigma_{y_i}}{\left| y_{i,\max} - y_{i,\min} \right|} + J_u
\end{equation}
Notice that $J_w$ depends on the controller structure and the tunable parameters therein, that is, $J_w=J_w(C(\vartheta))$ where $C$ is some controller and $\vartheta$ the related parameters.

In robust control theory, all stabilizing controllers can be expressed using the so-called Youla parameterization (\cite{Vidyasagar2011}), which characterizes a set of controllers. Consequently, an optimal controller that minimizes a specific performance index can be identified by searching within this set. This concept is useful for comparing two control methods. Given a plant $G$ and its model $\hat G$, let $\mathcal{D}_{2{\rm{term}}}(G,\hat G)$ represent the set of all stabilizing controllers employing the two-term DMC (Eq. (\ref{eq23})), and $\mathcal{D}_{3{\rm{term}}}(G,\hat G)$ denote the set of all stabilizing controller employing the three-term DMC (Eq. (\ref{eq33})). Defining the controllers of the two-term and three-term DMC as $C_{{\rm{2term}}}(\vartheta_2)$ and $C_{{\rm{3term}}}(\vartheta_3)$, respectively, the following theorem holds:

\begin{thm}\label{thm3}
Given a plant $G$ and its model $\hat G$, the three-term DMC controller outperforms the two-term DMC controller in the sense of $J_w$:
\begin{equation}\label{eq34}
  \min_{C_{{\rm{3term}}}(\vartheta_3)\in\mathcal{D}_{3{\rm{term}}}} \left\{J_w\left(C_{{\rm{3term}}}(\vartheta_3)\right)\right\} \leq \min_{C_{{\rm{2term}}}(\vartheta_2)\in\mathcal{D}_{2{\rm{term}}}} \left\{J_w\left(C_{{\rm{2term}}}(\vartheta_2)\right)\right\}.
  \end{equation}
\end{thm}

\begin{proof}
  Notice that $\dim\left(\vartheta_3\right) > \dim\left(\vartheta_2\right)$ and every stabilizing controller in the two-term DMC can be reproduced in the three-term QP by setting $\bm{S}$ as a zero matrix. This means that ${\mathcal{D}_{{\rm{3term}}}}(G,\hat G) \supset {\mathcal{D}_{2{\rm{term}}}}(G,\hat G)$. The three-term DMC searches for the optimal controller in a larger set than the two-term DMC does, therefore the inequality Eq. (\ref{eq34}) holds.
\end{proof}

\subsection{Simulation study: performance comparison without constraints}
The performance of two-term DMC and three-term DMC will be compared using a 2-input 2-output system. There are delays of 10 samples in the transfer function of the first output and delays of 2 samples in the second output. These delays will enlarge the difference between the two controller methods. Fig. \ref{fig: process C} shows the step response of the process. Unmeasured disturbances, $v_1(k)$ and $v_2(k)$, are present at the outputs and can be expressed as:

\begin{equation}
\begin{gathered}
v_i(k) = \frac{1 + 0.23q^{-1}}{1 - 0.9q^{-1}} e_i(k)
\end{gathered}
\end{equation}
The disturbances $v_1(k)$ and $v_2(k)$ are mutually independent white noises, $e_1(k)$ and $e_2(k)$, with zero means and a variance of 0.01. The performance index is determined based on the standard deviation of both the output and the control action. Identical horizon settings are employed for the two-term DMC and three-term DMC cases: horizon of dynamics $N = 55$, prediction horizon $P = 45$, and control horizon $M = 10$. The weighting parameters for each case are as follows:

\textbf{Case 1: Two-term DMC}
\newline $r_i=1$, $q_1$ and $q_2$ vary simultaneously from 0.01 to 1000, $\lambda_1=2$, $\lambda_2=1$;

\textbf{Case 2: Three-term DMC}
\newline $r_i=1$, $q_1$ and $q_2$ vary simultaneously from 0.01 to 1000, $s_1=4q_1$, $s_2=q_2$;

The performance and robustness of the two-term DMC and three-term DMC are compared. The weighting values are set using a rule to achieve the same closed-loop response time in both cases. Closed-loop simulations are performed to evaluate the ability of disturbance rejection. The standard deviations of the output and input are shown in Fig. \ref{fig: comparison1}, which demonstrates that the three-term DMC achieves smaller control errors and smaller control actions than the two-term DMC does. To demonstrate the robustness of the three-term DMC, simulations are carried out with model mismatch. Fig. \ref{fig: comparison2} shows that the closed-loop system becomes unstable for the two-term DMC, while the three-term DMC remains stable. Fig. \ref{fig: comparison3} shows that the two methods have nearly the same control error, but the three-term DMC has a smaller control action when the gain of the process decreases. These results indicate that the three-term DMC achieves higher control performance than the two-term DMC does and has a larger tolerance for model mismatch.
\begin{figure}
\begin{center}
\includegraphics[width=0.45\textwidth]{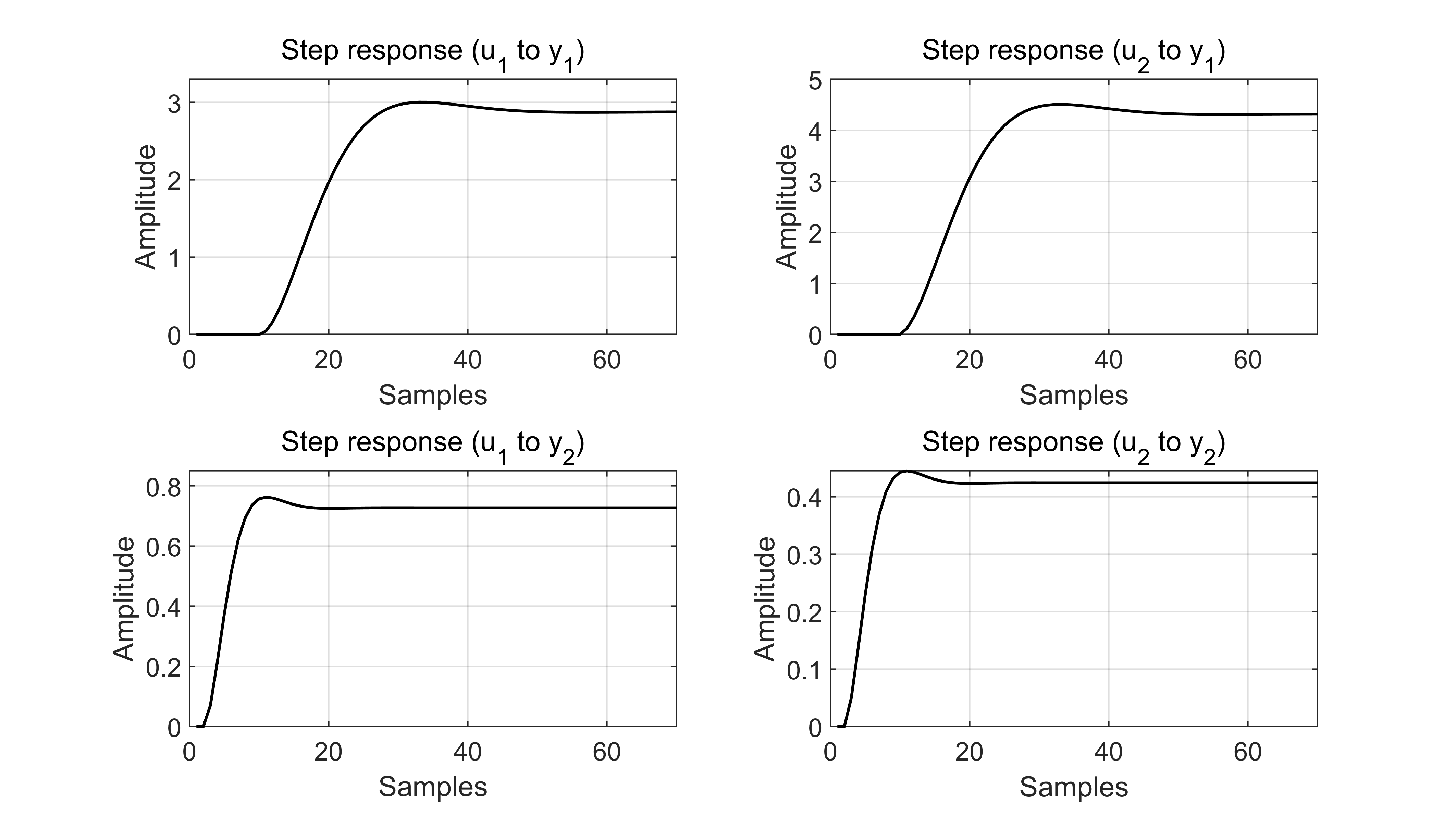}
\caption{Process C}
\label{fig: process C}
\end{center}
\end{figure}
\begin{figure}
\begin{center}
\includegraphics[width=0.45\textwidth]{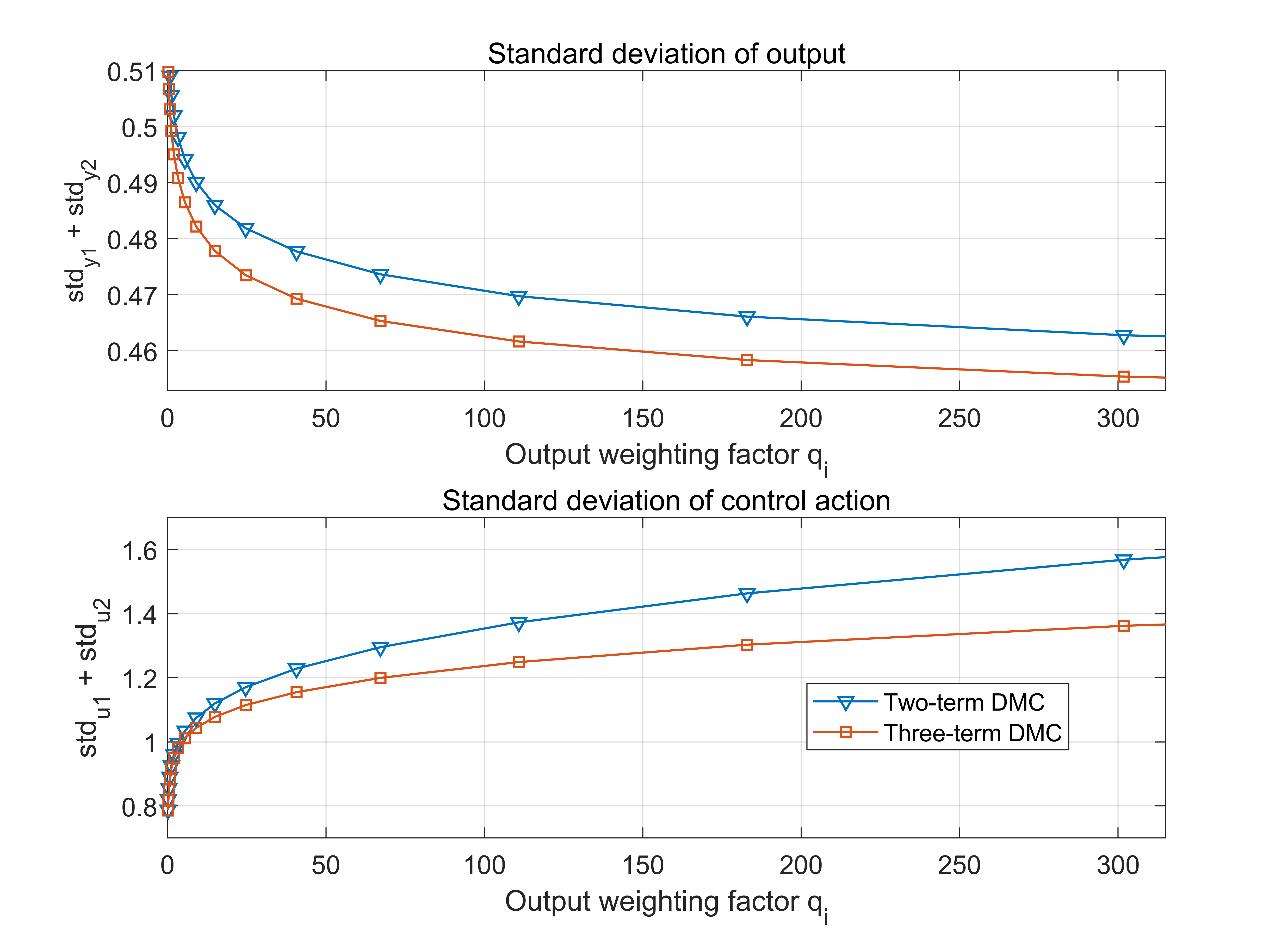}
\caption{Performance comparison of process C, no model mismatch}
\label{fig: comparison1}
\end{center}
\end{figure}
\begin{figure}
\begin{center}
\includegraphics[width=0.45\textwidth]{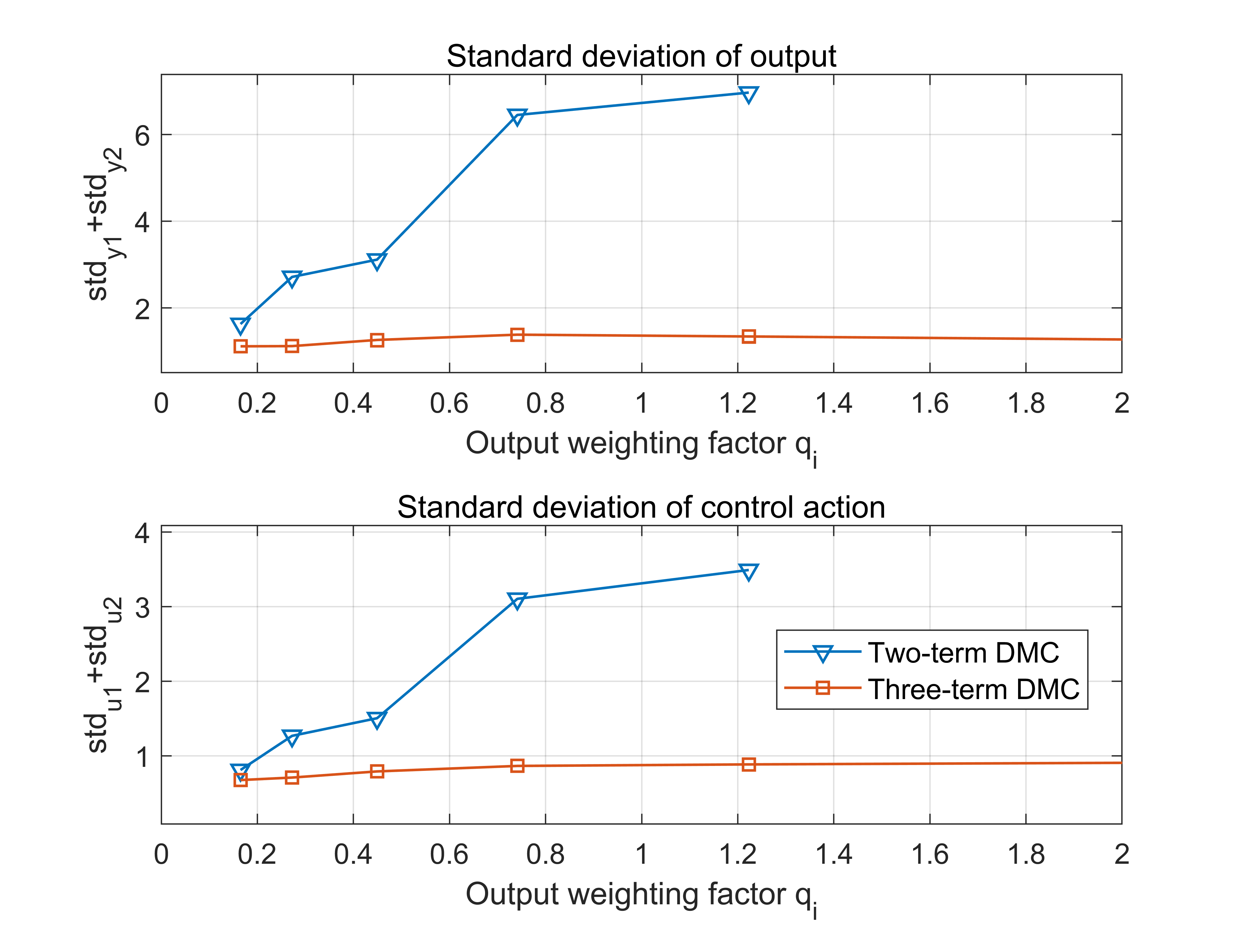}
\caption{Performance comparison of process C, process's gain = 2*Model's gain}
\label{fig: comparison2}
\end{center}
\end{figure}
\begin{figure}
\begin{center}
\includegraphics[width=0.45\textwidth]{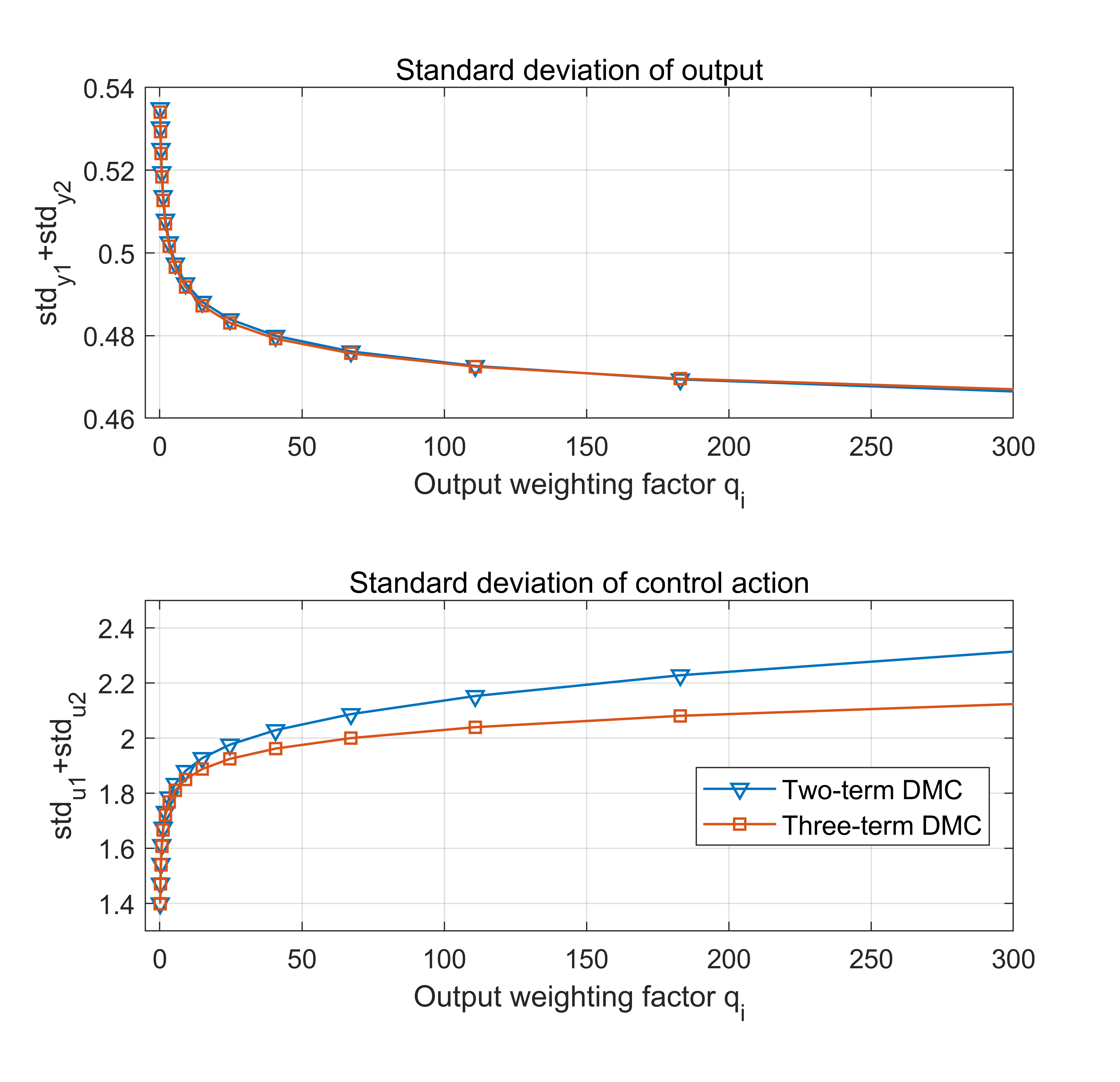}
\caption{Performance comparison of process C, process's gain = 0.5*Model's gain}
\label{fig: comparison3}
\end{center}
\end{figure}
\subsection{Simulation study: performance comparison with constraints}
MPC is widely acknowledged for its ability to handle constraints optimally. In practice, one of the common constraints are input saturation. Extensive researches have been devoted to the design of anti-windup schemes that mitigate the effects of control saturation (\citep{Zheng1998, Mulder2000, Adegbege2001}). A benchmark example which is commonly used in the study of anti-windup scheme (\citep{Zheng1998}) is used here. The process is shown in Fig. \ref{fig: process_morari}.
\begin{figure}
\begin{center}
\includegraphics[width=0.45\textwidth]{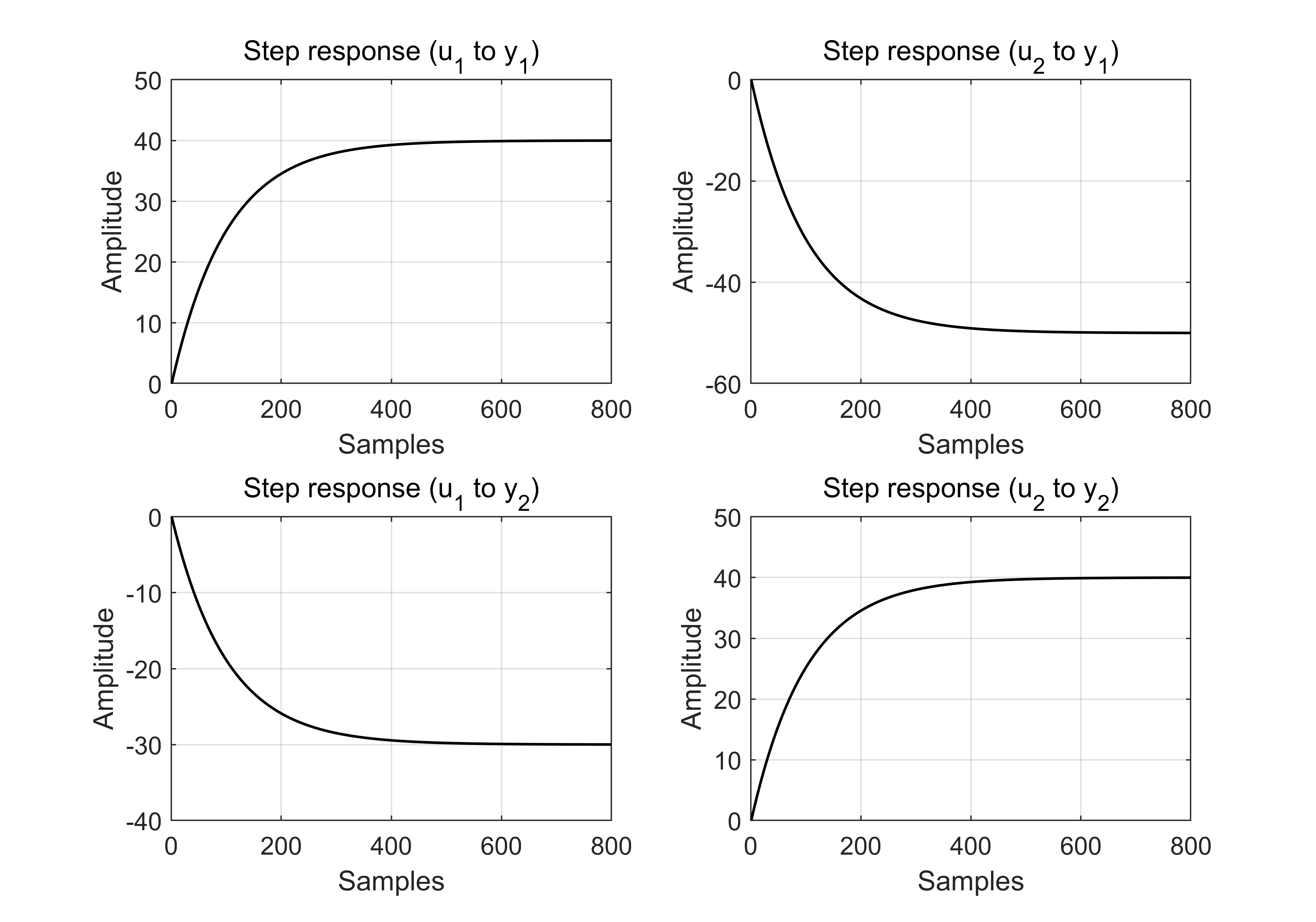}
\caption{Process D}
\label{fig: process_morari}
\end{center}
\end{figure}

In the simulations, two cases are compared as follows:

\textbf{Case 1: Two-term DMC}
\newline $r_j=10$, $q_i=1$, $\lambda_1=20$, $\lambda_2=20$; $ - 0.7 \le {u_j(k)} \le 0.7$;

\textbf{Case 2: Three-term DMC}
\newline $r_j=10$, $q_i=1$, $s_1=400q_1$, $s_2=400q_2$; $ - 0.7 \le {u_j(k)} \le 0.7$;

The simulation was conducted with the same horizon settings for both the two-term DMC and three-term DMC cases, namely horizon of dynamics of $N=500$, prediction horizon of $P=400$, and control horizon of $M=30$. In contrast to previous works in the literature \citep{Zheng1998, Mulder2000, Adegbege2001}, where the constraint $-1\le u_j(k) \le 1$ was applied, the constraint in this simulation was made more stringent to demonstrate the differences between the two controllers.

The simulation results indicate that both two-term and three-term DMC can successfully meet the stringent constraints. Notably, the response of the output from the three-term DMC is smoother compared to that of the two-term DMC. In case 2, the second output exhibits a slight reverse phenomenon. This difference in performance can be attributed to the non-diagonal form of the weighting matrix $\bm{Q}$ used in the three-term DMC.

\begin{figure}
\begin{center}
\includegraphics[width=0.45\textwidth]{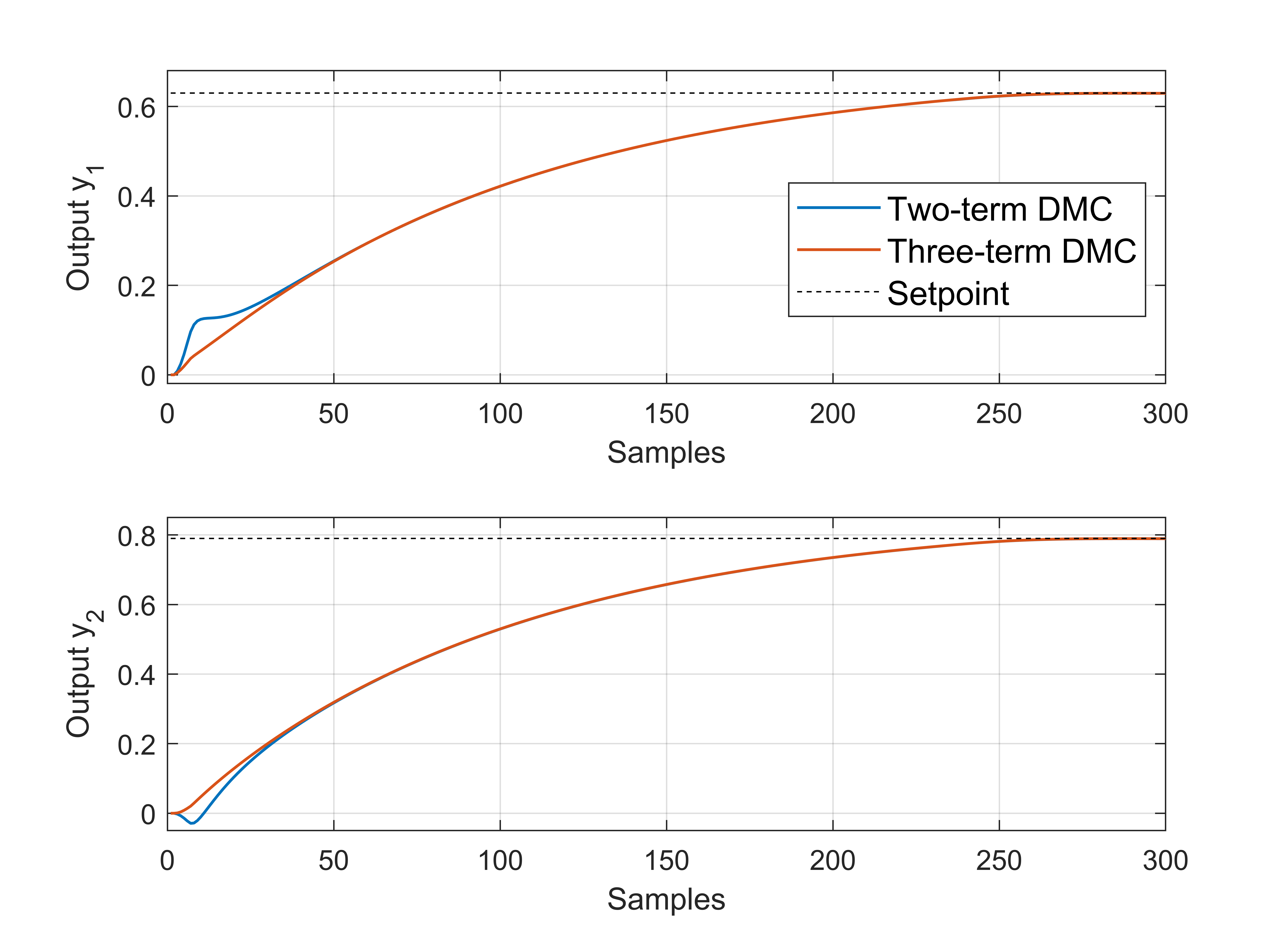}
\caption{Output comparison with constraint of process D}
\label{fig: limit_y}
\end{center}
\end{figure}
\begin{figure}
\begin{center}
\includegraphics[width=0.45\textwidth]{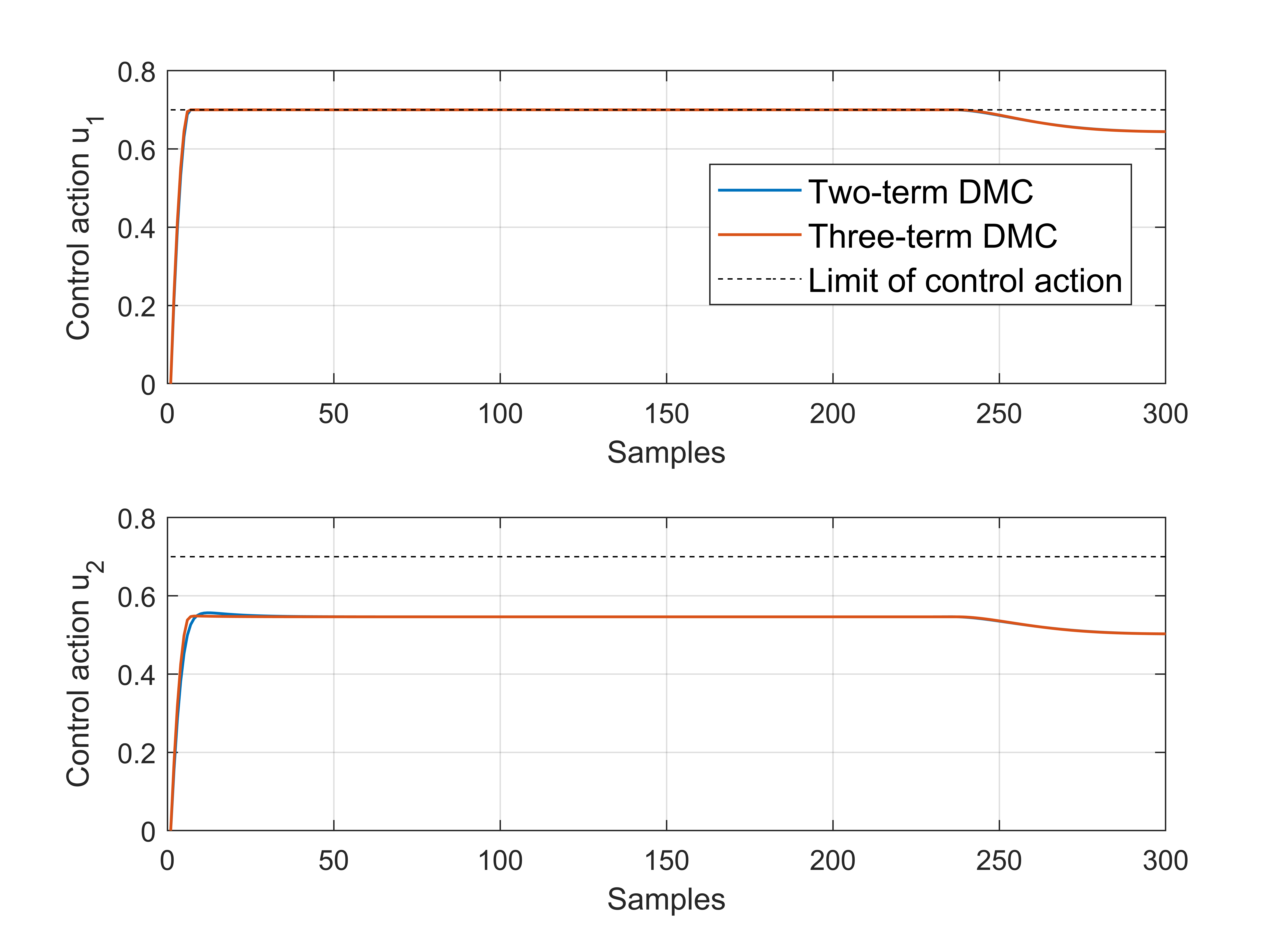}
\caption{Control action comparison with constraint of process D}
\label{fig: limit_u}
\end{center}
\end{figure}

\section{Conclusion}
This study investigates the three-term Dynamic Matrix Control (DMC) algorithm and establishes a relationship between the three-term DMC and the two-term DMC. Additionally, a formula for computing the ideal closed-loop response curves is derived. Based on the analysis, two controller tuning procedures are developed, one for closed-loop step response and one for disturbance reduction, which are simple, intuitive, and straightforward. Furthermore, a novel method for control method comparison is proposed, and using this approach, it is shown that the three-term DMC achieves higher performance and robustness than the traditional two-term DMC. The findings and tuning methods are supported by simulation results. This research is constructive and application-oriented. However, the stability of the three-term DMC remains an open question, and the authors welcome interested researchers to contribute further theoretical analysis of the three-term DMC. This work reveals that there is still room for improving controller performance and robustness, and one way to do so is to try new structures in control algorithms. Along this line of thinking, an interesting future research is to develop a three-term LGR(LQG) controller. Finally, an interesting finding in comparing two control methods is: "Control Method A is higher performing than Control Method B" is equivalent to "Control Method A is more robust than Control Method B". Note that when comparing two tunings of the same control method, Tuning A is higher performing than Tuning B often implies that Tuning A is less robust than Tuning B.

\bibliographystyle{elsarticle-num}
\bibliography{mybibliography}

\appendix

\section{Proof of Theorem \ref{thm1}}\label{proof_thm1}

The equivalent reference curve can be expressed as
\begin{equation}\label{.1}
{\bm{w}}_i^{'}(k) = {\left( {{\bm{I}} + {{\bm{T}}_4}{{\bm{T}}_2}} \right)^{ - 1}} \cdot \left( {{{\bm{r}}_i}(k) + {{\bm{T}}_4}{{\bm{T}}_3}{y_i}(k)} \right)
\end{equation}
where ${{\bm{r}}_i}(k)$ denotes the ${i^{th}}$ setpoint. ${y_i}(k)$ denotes the ${i^{th}}$ output value at last sample. ${q_i}$ and ${s_i}$ denote the output weighting coefficient and the output increment weighting coefficient regarding the ${i^{th}}$ output.

Through simple transforms, one can get
\begin{equation}\label{.2}
\left( {{\bm{I}} + {{\bm{T}}_4}{{\bm{T}}_2}} \right){\bm{w}}_i^{'}(k) = {\bm{r}_i}(k) + {{\bm{T}}_4}{{\bm{T}}_3}{y_i}(k)
\end{equation}
Notice that
\begin{equation}\label{.3}
{{\bm{T}}_2}{\bm{w}}_i^{'}(k) = \Delta {\bm{w}}_i^{'}(k) + {{\bm{T}}_3}w_i^{'}(k)
\end{equation}
where
\begin{equation}\label{.4}
\Delta {\bm{w}}_i^{'}(k) = {\left[ {w_i^{'}(k+1) - w_i^{'}(k),...,w_i^{'}(k + P) - w_i^{'}(k+P-1)} \right]^T}
\end{equation}
Eq. (\ref{.2}) can be rewritten as
\begin{equation}\label{.5}
{\bm{w}}_i^{'}(k) + {{\bm{T}}_4}\Delta {\bm{w}}_i^{'}(k) + {{\bm{T}}_4}{{\bm{T}}_3}w_i^{'}(k) = {{\bm{r}}_i}(k) + {{\bm{T}}_4}{{\bm{T}}_3}{y_i}(k)
\end{equation}
$w_i^{'}(k)$ is the start point of the reference curve ${\bm{w}}_i^{'}(k)$ and should be equal to $y_i (k)$, then the Eq. (\ref{.5}) can be simplified as
\begin{equation}\label{.6}
{\bm{w}}_i^{'}(k) + {{\bm{T}}_4}\Delta {\bm{w}}_i^{'}(k) = {{\bm{r}}_i}(k)
\end{equation}
Substitute ${{\bm{T}}_4} = {{\bm{Q}}^{ - 1}}{\bm{T}}_2^T{\bm{S}}$ into Eq. (\ref{.6}),
\begin{equation}\label{.7}
{\bm{w}}_i^{'}(k) + \frac{{{s_i}}}{{{q_i}}}{\bm{T}}_2^T\Delta {\bm{w}}_i^{'}(k) = {{\bm{r}}_i}(k)
\end{equation}
\begin{equation}\label{.8}
\frac{{{s_i}}}{{{q_i}}}\Delta {\bm{w}}_i^{'}(k) = {\bm{T}}_2^{ - T}\left( {{{\bm{r}}_i}(k) - {\bm{w}}_i^{'}(k)} \right)
\end{equation}
where
\begin{equation}\label{.9}
{\bm{T}}_2^{ - T} = {\left[ {\begin{array}{*{20}{c}}
  1&1& \cdots &1&1 \\
  {}&1& \cdots &1&1 \\
  {}&{}& \ddots & \vdots & \vdots  \\
  {}&{}&{}&1&1 \\
  {}&{}&{}&{}&1
\end{array}} \right]_{P \times P}}
\end{equation}
Regarding the vector ${\bm{w}}_i^{'}(k)$ as a function of time $k$, Eq. (\ref{.8}) can be expressed as
\begin{equation}\label{.10}
\frac{{{s_i}}}{{{q_i}}}\Delta w_i^{'}(k + h) = \sum\limits_{h = P}^0 {\left( {{r_i}(k) - w_i^{'}(k + h)} \right)}
\end{equation}
where $h = 1,...,P$. Assume the prediction horizon $P \to \infty $, and the sampling time ${T_s} \to 0$. One can get a continuous-time equation,
\begin{equation}\label{.11}
\frac{{{s_i}}}{{{q_i}}}\dot w_i^{'}(k + h) = \int\limits_\infty ^0 {\left( {{r_i}(k) - w_i^{'}(k + h)} \right)} dk
\end{equation}
Differentiate on both sides of the Eq. (\ref{.11}), one can get
\begin{equation}\label{.12}
\frac{{{s_i}}}{{{q_i}}}\ddot w_i^{'}(k + h) = w_i^{'}(k + h) - {r_i}(k)
\end{equation}
\section{Proof of Corollary \ref{cor1} and \ref{cor2}}\label{proof_cor12}
\subsection{Derivation of constant sequence}
Consider the reference signal is constant sequence, then the general solution of Eq. (\ref{.12}) is
\begin{equation}\label{.13}
w_i^{'}(k + h) = {C_1}{e^{h\sqrt {{q_i}/{s_i}} }} + {C_2}{e^{ - h\sqrt {{q_i}/{s_i}} }}
\end{equation}
and $w_i^{'}(k + h) = {r_i}(k)$ is a particular solution. So that the total solution is
\begin{equation}\label{.14}
w_i^{'}(k + h) = {r_i}(k) + {C_1}{e^{h\sqrt {{q_i}/{s_i}} }} + {C_2}{e^{ - h\sqrt {{q_i}/{s_i}} }}
\end{equation}
Under the assumption $w_i^{'}(k) = {y_i}(k)$, we have condition:
\begin{equation}\label{.15}
\left\{ \begin{gathered}
  w_i^{'}(k) = {y_i}(k) \hfill \\
  w_i^{'}(\infty ) = {r_i}(k) \hfill \\
\end{gathered}  \right.
\end{equation}
It can be easily obtained that ${C_1} = 0,{C_2} = {y_i}(k) - {r_i}(k)$.
To sum up,
\begin{equation}\label{.16}
w_i^{'}(k + h) = {r_i}(k) + \left( {{y_i}(k) - {r_i}(k)} \right){e^{ - h\sqrt {{q_i}/{s_i}} }}
\end{equation}

\begin{equation}\label{.17}
w_i^{'}(k + h) = {y_i}(k) + \left( {{r_i}(k) - {y_i}(k)} \right)\left( {1 - {e^{\frac{{ - h}}{{\sqrt {{s_i}/{q_i}} }}}}} \right)
\end{equation}
Eq. (\ref{.17}) is exactly a first order response curve. One can obtain:

\begin{equation}\label{.18}
{\lambda _i} = \sqrt {{s_i}/{q_i}}
\end{equation}

This equivalence relation shows that ${\bm{w}}_i^{'}(k)$ plays the role of reference trajectory, specifically a first order reference trajectory which is widely adopted. This property is very useful for controller tuning. For multi-variable system, ${s_i}$ can be designed independently for each output. Engineers only need to enter the expected closed-loop response time of every output, then the output increment weighting matrix $\bm{S}$ can be easily calculated.
\subsection{Derivation of ramp sequence}
The proof follows Eq. (\ref{.1}$\sim$\ref{.12}) in the proof of Theorem \ref{thm1}. Consider the reference signal is ramp sequence, one can also obtain a general solution as follows,
\begin{equation}
w_i^{'}(k + h) = {C_1}{e^{h\sqrt {{q_i}/{s_i}} }} + {C_2}{e^{ - h\sqrt {{q_i}/{s_i}} }}
\end{equation}
having the assumption that ${{{w}}_i}(k + h)$ is a ramp function, we know that ${{\ddot w}_i}(k + h) = 0$, then $w_i^{'}(k + h) = {w_i}(k + h)$ is a particular solution. So that the total solution is
\begin{equation}
w_i^{'}(k + h) = {w_i}(k + h) + {C_1}{e^{h\sqrt {{q_i}/{s_i}} }} + {C_2}{e^{ - h\sqrt {{q_i}/{s_i}} }}
\end{equation}
we have the solution condition:
\begin{equation}
\left\{ \begin{gathered}
  w_i^{'}(k) = {y_i}(k) \hfill \\
  w_i^{'}(k + P) = {w_i}(k + P) \hfill \\
\end{gathered}  \right.
\end{equation}
assuming $P$ is large enough, then we have ${C_1} = 0,{C_2} = {y_i}(k) - {w_i}(k)$.
To sum up,
\begin{equation}
w_i^{'}(k + h) = {w_i}(k + h) + \left( {{y_i}(k) - {w_i}(k)} \right){e^{ - h\sqrt {{q_i}/{s_i}} }}
\end{equation}

\section{Proof of Theorem \ref{thm2}}\label{proof_thm2}
\subsection{Derivation when $d_i = 1$}
In order to simplify the derivation, the single variable system will be discussed first; then, the result will be extended to the multi-variable system.
Consider a single variable process with one sample delays. In the three-term DMC scheme, weighting matrix $\bm{Q}$ and $\bm{S}$ has the structure as follows:
\begin{equation}
{\bf{Q}} = {\left[ {\begin{array}{*{20}{c}}
  q&{}&{} \\
  {}& \ddots &{} \\
  {}&{}&q
\end{array}} \right]_{P \times P}},{\bf{S}} = {\left[ {\begin{array}{*{20}{c}}
  s&{}&{} \\
  {}& \ddots &{} \\
  {}&{}&s
\end{array}} \right]_{P \times P}}
\end{equation}
The process has 1 sample delays, but in the diagonal of $\bm{Q}$ and $\bm{S}$, there is no zero because the model predicts the output from $k+1$ to $k+P$ at time $k$. The equivalent weighting matrix $\bm{Q}^{'}$ is as follows:
\begin{equation}
\begin{gathered}
  {{\bf{Q}}^{'}} = {\bf{Q}}\left( {{\bf{I}} + {{\bf{T}}_4}{{\bf{T}}_2}} \right) \hfill \\
   = {\left[ {\begin{array}{*{20}{c}}
  {q + 2s}&{ - s}&{}&{}&{} \\
  { - s}&{q + 2s}&{}&{}&{} \\
  {}&{}& \ddots &{}&{} \\
  {}&{}&{}&{q + 2s}&{ - s} \\
  {}&{}&{}&{ - s}&{q + s}
\end{array}} \right]_{P \times P}} \hfill \\
\end{gathered}
\end{equation}
Ignore the weighted norm of input increment ($\bm{R} = \bm{0}$), and rewrite the lost function Eq. (\ref{3to2})
\begin{equation}\label{J0_d=1}
\mathop {\min }\limits_{\Delta {\bf{U}}(k)} {J_{3term}}(k) = \left\| {\Delta {{\bf{e}}_P}(k)} \right\|_{{{\bf{Q}}^{'}}}^2
\end{equation}
where
\begin{equation}
{{\bf{e}}_P}(k) = {\bf{w^{'}}}(k) - {{\bf{y}}_P}(k) = {\bf{w}^{'}}(k) - \left( {{{\bf{y}}_{P0}}(k) + {\bf{A}}\Delta {\bf{u}}} \right)
\end{equation}
${{\bf{e}}_P}(k)$ is the predictive error from sample time $k+1$ to $k+P$,
\begin{equation}
{{\bf{e}}_P}(k) = \left[ {\begin{array}{*{20}{c}}
{e(k + 1)}&{e(k + 2)}& \ldots &{e(k + P)}
\end{array}} \right]
\end{equation}
${\bf{w}^{'}}(k)$ is the vector of the future set point,
\begin{equation}
{\bf{w}^{'}}(k) = {\left[ {\begin{array}{*{20}{c}}
{w^{'}(k + 1)}& \cdots &{w^{'}(k + P)}
\end{array}} \right]^T}
\end{equation}
To minimize lost function Eq. (\ref{J0_d=1}), taking the derivative of ${J_{3term}}(k)$ with respect to $\Delta {\bf{u}}(k)$, one obtains
\begin{equation}\label{differential}
\frac{{d{J_{3term}}(k)}}{{d\Delta {\bf{u}}(k)}} =  - 2{{\bf{A}}^T}{{\bf{Q}}^{'}}{{\bf{e}}_P}(k)
\end{equation}
the solution ${\Delta {\bf{u}}(k)}$ makes
\begin{equation}\label{solution_e1}
{\bf{e}}_P^*(k) = \left[ {\begin{array}{*{20}{c}}
  0&0& \ldots &0
\end{array}} \right]_{P \times 1}^T
\end{equation}
Eq. (\ref{solution_e1}) can be easily achieved theoretically when applying a one-step prediction strategy where $P=M$, resulting in a square and reversible dynamics matrix. Even when $P>M$, it is still acceptable due to the acceptable level of error tolerance.
then we have
\begin{equation}
y(k + h) = {w^{'}}(k + h)
\end{equation}
Similarly, in the multi-variable system, it's easy to obtain the optimal solution,
\begin{equation}
{\bf{e}}_P^*(k) = \left[ {\begin{array}{*{20}{c}}
  0&0& \ldots &0
\end{array}} \right]_{pP \times 1}^T
\end{equation}
accordingly, for each output,
\begin{equation}
y_i(k + h) = {w_i^{'}}(k + h)
\end{equation}
\subsection{Derivation when ${d_i} > 1$}
Consider a single variable process with two sample delays. In the three-term DMC scheme, weighting matrix $\bm{Q}$ and $\bm{S}$ has the structure as follows:
\begin{equation}
{\bf{Q}} = {\left[ {\begin{array}{*{20}{c}}
0&{}&{}&{}\\
{}&q&{}&{}\\
{}&{}& \ddots &{}\\
{}&{}&{}&q
\end{array}} \right]_{P \times P}},{\bf{S}} = {\left[ {\begin{array}{*{20}{c}}
0&{}&{}&{}\\
{}&s&{}&{}\\
{}&{}& \ddots &{}\\
{}&{}&{}&s
\end{array}} \right]_{P \times P}}
\end{equation}
The process has 2 sample delays, but in the diagonal of $\bm{Q}$ and $\bm{S}$, there is only one zero because the model predicts the output from $k+1$ to $k+P$ at time $k$. The equivalent weighting matrix $\bm{Q}^{'}$ is as follows:
\begin{equation}
\begin{array}{l}
{{\bf{Q}}^{'}} = {\bf{Q}}\left( {{\bf{I}} + {{\bf{T}}_4}{{\bf{T}}_2}} \right)\\
 = {\left[ {\begin{array}{*{20}{c}}
s&{ - s}&{}&{}&{}&{}\\
{ - s}&{q + 2s}&{ - s}&{}&{}&{}\\
{}&{ - s}&{q + 2s}&{}&{}&{}\\
{}&{}&{}& \ddots &{}&{}\\
{}&{}&{}&{}&{q + 2s}&{ - s}\\
{}&{}&{}&{}&{ - s}&{q + s}
\end{array}} \right]_{P \times P}}
\end{array}
\end{equation}
Ignore the weighted norm of input increment ($\bm{R} = \bm{0}$), and rewrite the lost function Eq. (\ref{3to2})
\begin{equation}\label{J0_d>1}
\begin{gathered}
\mathop {\min }\limits_{\Delta {\bf{U}}(k)} {J_{3term}}(k) = \left\| {\Delta {{\bf{e}}_P}(k)} \right\|_{{{\bf{Q}}^{'}}}^2
\end{gathered}
\end{equation}
where
\begin{equation}
{{\bf{e}}_P}(k) = {\bf{w^{'}}}(k) - {{\bf{y}}_P}(k) = {\bf{w}^{'}}(k) - \left( {{{\bf{y}}_{P0}}(k) + {\bf{A}}\Delta {\bf{u}}} \right)
\end{equation}
${{\bf{e}}_P}(k)$ is the predictive error from sample time $k+1$ to $k+P$,
\begin{equation}
{{\bf{e}}_P}(k) = {\left[ {\begin{array}{*{20}{c}}
  {e(k + 1)}&{e(k + 1)}& \ldots &{e(k + P)}
\end{array}} \right]^T}
\end{equation}
${\bf{w}^{'}}(k)$ is the vector of the future set point,
\begin{equation}
{\bf{w}^{'}}(k) = {\left[ {\begin{array}{*{20}{c}}
{w^{'}(k + 1)}& \cdots &{w^{'}(k + P)}
\end{array}} \right]^T}
\end{equation}

To minimize lost function Eq. (\ref{J0_d>1}), taking the derivative of ${J_{3term}}(k)$ with respect to $\Delta {\bf{u}}(k)$, one obtains
\begin{equation}\label{differential_d>1}
\frac{{d{J_{3term}}(k)}}{{d\Delta {\bf{u}}(k)}} =  - 2{{\bf{A}}^T}{{\bf{Q}}^{'}}{{\bf{e}}_P}(k)
\end{equation}
where
\begin{equation}
{\bf{A}} = {\left[ {\begin{array}{*{20}{c}}
0& \cdots &0\\
{a(1)}&{}&0\\
 \vdots & \ddots & \vdots \\
{a(M)}& \cdots &{a(1)}\\
 \vdots &{}& \vdots \\
{a(P - 1)}& \cdots &{a(P - M)}
\end{array}} \right]_{P \times M}}
\end{equation}
$a(k)$ is the step response coefficients, $a(1) \ne 0$. Suppose $\Delta {{\bf{u}}^0}(k)$ is the solution which makes
\begin{equation}\label{e0}
{\bf{e}}_P^0(k) = {\left[ {\begin{array}{*{20}{c}}{w^{'}(k + 1)}&0& \ldots &0\end{array}} \right]^T}
\end{equation}
There is a counterintuitive result:
$\Delta {{\bf{u}}^0}(k)$ is not the optimal solution of the proposition Eq. (\ref{J0_d>1}). Substitute
Eq. (\ref{e0}) into Eq. (\ref{differential}), one obtains
\begin{equation}
\begin{array}{l}
{\left. {\frac{{d{J_{3term}}(k)}}{{d\Delta {\bf{u}}(k)}}} \right|_{\Delta {\bf{u}}(k) = \Delta {{\bf{u}}^0}(k)}} =  - 2{{\bf{A}}^T}{{\bf{Q}}^{'}}{\bf{e}}_P^0(t)\\
 - 2\left[ {\begin{array}{*{20}{c}}
0&{a(1)}& \cdots &{a(M)}& \cdots &{a(P - 1)}\\
 \vdots &{}& \ddots &{}&{}& \vdots \\
0&0& \cdots &{a(1)}& \cdots &{a(P - M)}
\end{array}} \right]\left[ {\begin{array}{*{20}{c}}
{s{w^{'}}(t + 1)}\\
{ - s{w^{'}}(t + 1)}\\
0\\
 \vdots \\
0
\end{array}} \right]\\
 =  - 2\left[ {\begin{array}{*{20}{c}}
{ - a(1)s{w^{'}}(t + 1)}\\
0\\
0\\
 \vdots \\
0
\end{array}} \right]
\end{array}
\end{equation}
${\left. {\frac{{d{J_{3term}}(k)}}{{d\Delta {\bf{u}}(k)}}} \right|_{\Delta {\bf{u}}(k) = \Delta {{\bf{u}}^0}(k)}} \ne {\bf{0}}$, which means ${\Delta {\bf{u}}(k) = \Delta {{\bf{u}}^0}(k)}$ is not the optimal solution.

One should focus on ${{\bf{Q}}^{'}}{{\bf{e}}_P}(k)$. Suppose $\Delta {{\bf{u}}^*}(k)$ is the solution which makes a predictive error ${\bf{e}}_P^*(k)$ as follows
\begin{equation}\label{e_p^*}
{\bf{e}}_P^*(k) = {\left[ {\begin{array}{*{20}{c}}
{e(k + 1)}& \ldots &{e(k + P)}
\end{array}} \right]^T}
\end{equation}
If ${\bf{e}}_P^*(k)$ satisfies the following condition, one can obtain $\frac{{d{J_{3term}}(k)}}{{d\Delta {\bf{u}}(k)}} = {\bf{0}}$.
\begin{equation}\label{condition}
(q + 2s)e(k + h) = se(k + h - 1) + se(k + h + 1)
\end{equation}
where $h = 1, \ldots P$. Suppose ${\bf{e}}_P^*(k)$ satisfies Eq. (\ref{condition}), then we have
\begin{equation}\label{QE1}
{{\bf{Q}}^{'}}{\bf{e}}_P^*(k) = \left[ {\begin{array}{*{20}{c}}
{se(k + 1) - se(k + 2)}\\
0\\
0\\
 \vdots \\
{se(k + P + 1) - se(k + P)}
\end{array}} \right]
\end{equation}
Assumption \ref{ass1} ensures that $e(k + P) - e(k + P + 1) = 0$, then Eq. (\ref{QE1}) is simplified as
\begin{equation}\label{QE2}
{{\bf{Q}}^{'}}{\bf{e}}_P^*(k) = \left[ {\begin{array}{*{20}{c}}
{se(k + 1) - se(k + 2)}\\
0\\
0\\
 \vdots \\
0
\end{array}} \right]
\end{equation}
then one can check the derivative value,
\begin{equation}
{\left. {\frac{{d{J_{3term}}(k)}}{{d\Delta {\bf{u}}(k)}}} \right|_{\Delta {\bf{u}}(k) = \Delta {{\bf{u}}^*}(k)}} =  - 2{{\bf{A}}^T}{{\bf{Q}}^{'}}{\bf{e}}_P^*(t) = {\bf{0}}
\end{equation}
then the closed-loop response can be obtained if one solves Eq. (\ref{condition}).
Simplify Eq. (\ref{condition}) as follows
\begin{equation}\label{ch_equation1}
e(k + h - 1) - (\frac{q}{s} + 2)e(k + h) + e(k + h + 1) = 0
\end{equation}
Eq. (\ref{ch_equation1}) is a homogeneous differential equation with constant coefficients. The characteristic equation is
\begin{equation}\label{ch_equation2}
{\alpha ^2} - \left( {\frac{q}{s} + 2} \right)\alpha  + 1 = 0
\end{equation}
the form of the solution is
\begin{equation}\label{form}
e(k + h) = C{\alpha ^h}
\end{equation}
the solution of Eq. (\ref{ch_equation2}) is
${\alpha _1} = \frac{{\left( {\frac{q}{s} + 2} \right) - \sqrt {{{\left( {\frac{q}{s} + 2} \right)}^2} - 4} }}{2}$ and ${\alpha _2} = \frac{{\left( {\frac{q}{s} + 2} \right) + \sqrt {{{\left( {\frac{q}{s} + 2} \right)}^2} - 4} }}{2}$. $0 < {\alpha _1} < 1$ and ${\alpha _2} > 1$. Then we have
\begin{equation}
e(k + h) = {C_1}\alpha _1^h + {C_2}\alpha _2^h
\end{equation}
$e(k + h)$ should be convergent which means ${C_2} = 0$. Considering that the process has two sample delays, we have the condition:
\begin{equation}\label{solu_con1}
e(k + 1) = w(k + 1)
\end{equation}
It is easy to obtain ${C_1} = w^{'}(k + 1)\alpha _1^{ - 1}$. The expression of $e(k)$ is
\begin{equation}
e(k + h) = w^{'}(k + 1)\alpha _1^{h - 1}
\end{equation}
If the process's delay is larger than two samples, the only thing changed is the solution condition Eq. (\ref{solu_con1}). Denote the delay of the process as $d$ samples, $d>1$, then
\begin{equation}\label{solu_con2}
e(k + d - 1) = w^{'}(k + d - 1) = {C_1}\alpha _1^{d - 1}
\end{equation}
then one has ${C_1} = w^{'}(k + d - 1)\alpha _1^{ - d + 1}$. The expression of $e(k)$ can be formulated as
\begin{equation}\label{e_**}
e(k + h) = \left\{ {\begin{array}{*{20}{c}}
  {w^{'}(k + d - 1){\quad \quad \quad \quad \quad},h \leqslant d - 1} \\
  {w^{'}(k + d - 1)\alpha _1^{h - d + 1}{\quad \quad},h > d - 1}
\end{array}} \right.
\end{equation}
Having the expression of ${e}(k + h)$, the predictive output can be get through ${{\bf{y}}_P}(k) = {{\bf{w}}^{'}}(k) - {\bf{e}}_P^*(k)$. Namely,
\begin{equation}\label{y_cl}
{y_P}(k + h) = \left\{ {\begin{array}{*{20}{c}}
  {0{\quad \quad \quad \quad \quad \quad},h \leqslant d - 1} \\
  {w^{'}(k + h) - w^{'}(k + d - 1)\alpha _1^{h - d + 1},h > d - 1}
\end{array}} \right.
\end{equation}

The calculation gives a best trajectory of the predictive error ${{\bf{w^{'}}}(k) - {{\bf{y}}_P}(k)}$ which minimize the weighted term ${\left\| {{\bf{w}^{'}}(k) - {{\bf{y}}_P}(k)} \right\|_{{{\bf{Q}}^{'}}}}$. Therefore, every output can be considered independently. Then one obtains the following result.

The estimation of the $i^{th}$ output is
\begin{equation}\label{inprove_est_multi}
{y_i}(k + h) = \left\{ {\begin{array}{*{20}{c}}
{0{\quad \quad \quad \quad \quad \quad},h \le {d_i} - 1}\\
{{w_i^{'}}(k + h) - {w_i^{'}}(k + {d_i} - 1)\alpha _i^{h - d + 1}{\rm{   }},h > {d_i} - 1}
\end{array}} \right.
\end{equation}
where $h = 1, \ldots P$,
\begin{equation}
{\alpha_i} = \frac{{\left( {\frac{q_i}{s_i} + 2} \right) - \sqrt {{{\left( {\frac{q_i}{s_i} + 2} \right)}^2} - 4} }}{2}
\end{equation}


\end{document}